\documentclass[lettersize,journal]{IEEEtran}
\usepackage{enumerate}
\usepackage{amsmath}
\usepackage{amsfonts}
\usepackage{amssymb}
\usepackage{amsthm}
\usepackage{algorithmic}
\usepackage{array}
\usepackage[caption=false,font=normalsize,labelfont=sf,textfont=sf]{subfig}
\usepackage{textcomp}
\usepackage{stfloats}
\usepackage{url}
\usepackage{verbatim}
\usepackage{graphicx}
\usepackage{multirow}
\usepackage{color}

\def\BibTeX{{\rm B\kern-.05em{\sc i\kern-.025em b}\kern-.08em
    T\kern-.1667em\lower.7ex\hbox{E}\kern-.125emX}}
\usepackage{balance}

\newcommand{\QED}{$\square\,$}
\theoremstyle{definition}
\newtheorem{definition}{Definition}
\newtheorem{theorem}{Theorem}
\newtheorem{lemma}{Lemma}
\newtheorem{protocol}{Protocol}
\newtheorem{property}{Property}
\newtheorem{proposition}{Proposition}
\renewenvironment{proof}[1][Proof: ]{\vspace{7pt}\noindent{\bf #1}\small}{\QED}

\newcommand{\bra}[1]{ \left\lvert #1\right\rangle}
\newcommand{\ket}[1]{ \left\langle #1\right\lvert}

\newcommand{\innerpro}[2]{ \left\langle #1\lvert #2\right\rangle}
\newcommand{\outerpro}[2]{ \bra{#1}\ket{#2}}
\newcommand{\floor}[1]{\left \lfloor #1 \right \rfloor}
\newcommand{\ceil}[1]{\left \lceil #1 \right \rceil}

\newcommand{\set}[1]{\left \{ #1 \right \}}
\newcommand{\intset}[1]{\left [ #1 \right ]}
\newcommand{\odd}[1]{\left [ #1 \right ]^o}
\newcommand{\cint}[1]{ 1,2,\cdots, #1}
\newcommand{\abs}[1]{\left \lvert#1\right\lvert}
\newcommand{\sfrac}[1]{\frac{1}{\sqrt{#1}}}

\newcommand{\SWAP}{\mathcal{SWAP}}
\newcommand{\QFT}{\mathcal{QFT}}
\newcommand{\QFIT}{\mathcal{QFT}^\dagger}

\newcommand{\PROT}{\mathcal{P}}
\newcommand{\CPROT}{\mathcal{CP}}
\newcommand{\Toffoli}{\mathcal{T}}
\newcommand{\BROT}{\mathcal{BROT}}

\newcommand{\ROT}{\mathcal{ROT}}
\newcommand{\MUL}{\mathcal{MUL}}
\newcommand{\BMUL}{\mathcal{BMUL}}
\newcommand{\SUM}{\mathcal{SUM}}
\newcommand{\BSUM}{\mathcal{BSUM}}
\newcommand{\XOR}{\mathcal{XOR}}

\newcommand{\CNOT}{\mathcal{CNOT}}
\renewcommand{\H}{\mathcal{H}}
\newcommand{\U}{\mathcal{U}}

\newcommand{\onarrow}[1]{\overset{#1}{\longrightarrow}}

\renewcommand{\mod}{\ \mathrm{mod}\ }

\begin{document}
\title{Secure and Efficient Two-party Quantum Scalar Product Protocol With Application to Privacy-preserving Matrix Multiplication}
\author{Wen-Jie Liu, Zi-Xian Li
\thanks{Manuscript created February 2, 2023; This work was developed by National Natural Science Foundation of China (62071240), the Innovation Program for Quantum Science and Technology (2021ZD0302900), and the Research Innovation Program for College Graduates of Jiangsu 
Province, China (KYCX23\_1370).}
\thanks{The authors are with the School of Software, Nanjing University of Information Science and Technology, Nanjing 210044, China (e-mail: wenjiel@163.com; zixianli157@163.com).}

}

\maketitle

\begin{abstract}
Secure two-party scalar product (S2SP) is a promising research area within secure multiparty computation (SMC), which can solve a range of SMC problems, such as intrusion detection, data analysis, and geometric computations. However, existing quantum S2SP protocols are not efficient enough, and the complexity is usually close to exponential level. In this paper, a novel secure two-party quantum scalar product (S2QSP) protocol based on Fourier entangled states is proposed to achieve higher efficiency. Firstly, the definition of unconditional security under malicious models is given. And then, an honesty verification method called Entanglement Bondage is proposed, which is used in conjunction with the modular summation gate to resist malicious attacks. The property of Fourier entangled states is used to calculate the scalar product with polynomial complexity. The unconditional security of our protocol is proved, which guarantees the privacy of all parties. In addition, we design a privacy-preserving quantum matrix multiplication protocol based on S2QSP protocol. By transforming matrix multiplication into a series of scalar product processes, the product of two private matrices is calculated without revealing any privacy. Finally, we show our protocol's feasibility in IBM Qiskit simulator. 
\end{abstract}

\section{Introduction}\label{sec1}
\IEEEPARstart{S}{ecure} Multi-party Computation (SMC) enables multiple parties who do not trust each other to collaboratively compute a target function using their respective private data, while preserving the privacy of all participants. Since Yao\cite{1982YaoPro} first proposed this concept in 1982, the main solutions for SMC proposed by classical cryptographers include Garbled Circuit\cite{1986YaoHow}, Oblivious Transfer\cite{1985EvenA,1987BrassardA,2008PeikertA}, Secret Sharing\cite{1987GoldreichHow,1988Ben-OrCom,1992BeaverEff,2009XiongEff}, Homomorphic Encryption\cite{1985ElgamalA,1999PaillierPub,2008IvanInt}, etc. However, protocols with higher generality typically exhibit greater complexity. As a result, researchers often focus on developing specialized SMC protocols tailored to specific problems. Secure Two-party Scalar Product (S2SP) is a research area that investigates how two parties can securely compute the scalar product of their respective private vectors. Many SMC problems, such as intrusion detection\cite{2016HuangSec,2018ZhuEff,2019ShenShe}, data analysis\cite{2001DuPri,2002VaidyaPri,2005GoethalsOn,2011TranAn}, and geometry computation\cite{2001AtallahSec,2009ThomasSec,2012YangSec}, can be reduced to S2SP, which makes it a crucial building block for general secure computation. However, current classical S2SP schemes either exhibit high complexity or rely on the Computational Hardness Assumption for their security. Thus, there is an urgent requirement for a novel scheme that achieves both high security and computational efficiency.

In recent years, there has been a growing interest in utilizing quantum mechanisms with potential for unconditional security to achieve secure multiparty computation (SMC), which is referred to as quantum SMC (QSMC)\cite{2018LiuPri,2019LiA,2021ShiSec,2022ShiQuaA,2022ShiQuaB}. However, existing Secure Two-party Quantum Scalar Product (S2QSP) protocols are not efficient enough. In 2012, He et al.\cite{2012HeA} pioneered a S2QSP protocol, which requires a non-colluding third party to distribute entangled states among the two participants. Furthermore, it demands significant entanglement resources, which may exceed the required level when operating on sparse private input vectors. In 2018, Wang et al.\cite{2018WangQua} proposed a new S2QSP scheme using classical cryptography and continuous-variable clusters. Their scheme no longer needs a third party, but still needs massive redundant quantum resources and measurement operations. In 2019, Shi et al.\cite{2019ShiStr} proposed a strong privacy-preserving S2QSP protocol using Grover's algorithm\cite{1997GroverQua} with constant communication complexity. However, while Grover's algorithm provides a quadratic speedup, its computational complexity remains close to exponential. As a result, existing S2QSP protocols only offer unconditional security with limited improvement in computational efficiency.

Compared with S2QSP, people have more fully studied two QSMC problems, i.e., Secure Multi-party Quantum Summation or Multiplication (SMQS or SMQM), where several parties can secretly add or multiply up their private integers. In 2013, Yang et al.\cite{2013YangSec} proposed a secret sharing scheme based on Quantum Fourier Transform $\QFT$. It takes advantage of a property of $d$-level cat state\cite{2002KarimipourEnt}, i.e., the ability to keep the sum unchanged after applying the Modular Summation Gate $\SUM$. Based on this scheme, people have proposed several SMQS protocols\cite{2018YangSec,2019JiQua,2020SutradharA,2020SutradharHyb,2021SutradharA,2021YiQua}. However, these protocols are limited to SMQS. In 2016, Shi et al.\cite{2016ShiSec} provided another way. They transformed the calculation from the bit domain to the phase domain, then used the Rotation Gate $\ROT$ to perform the addition. This idea is inspired by Draper's Transform Adder\cite{2000DraperAdd}. Here, SMQM is implemented by using the Modular Multiplication Gate $\MUL$. In addition, they used Bitwise XOR Gate $\XOR$ (or $\CNOT$ gate) to construct the $d$-level cat state, and applied again when returning to untie the entanglement, and check the honesty of the other party. Considering that the states appearing in the protocol have exceeded the category of cat states, we define them as a more generalized state called Fourier Entangled (FE) State.

Although SMQS and SMQM protocols using FE states are efficient, they cannot be directly applied to scalar product computations due to the interference between addition and multiplication in such calculations. However, we discover the following insight: Shi et al.'s method\cite{2016ShiSec} allows both addition and multiplication to be performed simultaneously without interference, as addition is performed in the phase domain and multiplication is performed in the bit domain. In this paper, we propose a novel secure quantum scalar product protocol based on this nature, to achieve higher computational efficiency. To transform this idea into a sufficiently secure S2QSP protocol, we first use $\SUM$ to impose several random number in the bit domain to prevent the measurement attack. What's more, we apply a bi-particle version of $\SUM$ ($\BSUM$) to entangle all particles, so as to prevent verify the honesty of the sender. We refer to this operation as Entanglement Bondage, and design an corresponding honesty test to check the sender's honesty. In this way, the two parties will get an almost fair position, and can resist against forgery attacks without the assistance of a third party. It is worth noting that since $\QFT$ is one of the few quantum algorithms that can achieve exponential acceleration, our S2QSP scheme can calculate scalar products with the highest known efficiency, namely polynomial complexity.

Our contributions in this paper are summarized below.
\begin{enumerate}[1)]
    \item Based on the conceptions of Leakage Degree and Negligibility, the definition of unconditional security under the malicious model is given in detail. 
    \item An honesty verification method called Entanglement Bondage is proposed, which is used in conjunction with the modular summation gate to resist various malicious attacks. 
    \item Based on the property of Fourier Entangled state and the methods in 2), we propose a S2QSP protocol with polynomial complexity.  and prove its unconditional security.
    \item Based on the proposed S2QSP protocol, we present a privacy-preserving matrix multiplication protocol, as an extended application of it.
    \item Finally, we verify the feasibility of the proposed S2QSP protocol in IBM Qiskit simulator.
\end{enumerate}

The rest of this paper is arranged as follows. In Section~\ref{sec2}, we define the notations we used, give the quantum operations to be used, and introduce the Fourier entangled state. In Section~\ref{sec3}, we define unconditional security in the malicious model, propose the entanglement bondage, then present our protocol. We analyze our protocol in Section~\ref{sec4}, and give an application of it in Section~\ref{sec5}. We conclude in Section~\ref{sec6}.

\section{Preliminary}\label{sec2}
\noindent 

\subsection{Definitions of Notations}\label{sec2.1}
\noindent
Table~\ref{table_sym} shows the definitions of notations used.

\begin{table}[h]
\begin{center}
\caption{Definitions of Notations}\label{table_sym}
\begin{tabular}{|c|c|}
\hline
Symbols & Meanings\\
\hline
$\imath$   & Imaginary unit  \\
\hline
$n$   & Dimension of input vector   \\
\hline
$\mathbf{x}=(x_1,x_2,\cdots,x_n)$   &  Alice's input vector \\
\hline
$\mathbf{y}=(y_1,y_2,\cdots,y_n)$   &  Bob's input vector \\
\hline
$v$   &  Bob's input random integer\\
\hline
$u$   &  Alice's output \\
\hline
$d,D=2^d$  & Particle's qubit number and dimension \\
\hline
$a^{-1}$  & Multiplicative inverse of $a \mod D$\\
\hline
$\oplus$ & Bitwise XOR\\
\hline
$\omega$  & $e^\frac{\imath 2\pi}{D}=\cos(\frac{2\pi}{D})+\imath \sin(\frac{2\pi}{D})$ \\ 
\hline
$m,N=2^m$  & Output's bit number and modulus\\
\hline
$\intset{N}$  & Set $\set{0,1,\cdots,N-1}$ \\
\hline
$\odd{N}$  & Set composed of odd numbers in $\intset{N}$ \\
\hline
$a\lvert b$   & $a$ is a factor of $b$ \\
\hline
$a\equiv b(\mod D)$  & $D|(a-b)$\\
\hline
$\bra{\psi}_h$  & Particle $h$ is in state $\bra{\psi}$ \\
\hline
$\U_h$  & Quantum gate $\U$ performed on particle $h$\\
\hline
$a\parallel b$   & String connected by two strings $a,b$\\
\hline
$(h,t)$  & System composed of particles $h$ and $t$\\
\hline
$\floor{a},\ceil{a}$   & Round integer $a$ up and down respectively\\
\hline
$\abs{A}$ &  Cardinality of set $A$\\
\hline
$g:A\to B$ & Function $g$ from set $A$ to set $B$\\
\hline
$g^{-1}(b)$   & Solution set $\set{a\lvert g(a)=b,a\in A}$\\
\hline
$Im(g)$   & Image set $\set{ b\lvert g(a)=b,a\in A}$\\
\hline
$\Pr(A)$   & Probability of event $A$\\
\hline
$H(A,B),H(A:B),$  & Joint entropy, mutual information and conditional\\
$H(A|B)$   & entropy of random variables $A,B$ \cite{2000NielsenQua}\\
\hline
$Z_A,Z_B$   & View of Alice, Bob respectively\\
\hline
$X_A,X_B$  & Privacy of Alice, Bob respectively\\
\hline
$I_A,I_B$   & Leakage degree of Alice, Bob's privacy respectively\\
\hline
\end{tabular}
\end{center}
\end{table}

\subsection{Quantum Operations}\label{sec2.2}
\noindent
Taking two $d$-qubits particles $h=\left(h_{d-1},h_{d-2},\cdots,h_{0}\right)$ and $t=\left(t_{d-1},t_{d-2},\cdots,t_0\right)$ as example, the quantum gates used in this paper are as follows (where addition and multiplication are all performed mod $D$).
\begin{enumerate}[1)]
    \item {Quantum Fourier Transform $\QFT$ and its inverse:
    \begin{align}
        &\QFT_h:\bra{a}_h\rightarrow \sfrac{D}\sum_{j\in\intset{D} }\omega^{aj}\bra{j}_h,\\
        &\QFIT_h:\bra{a}_h\rightarrow \sfrac{D}\sum_{j\in\intset{D} }\omega^{-aj}\bra{j}_h.
    \end{align}
    }
    \item {Rotation Gate $\ROT(b)$ where $b\in\intset{D}$:
    \begin{equation}
    \ROT(b)_h:\bra{a}_h\rightarrow \omega^{ab}\bra{a}_h.
    \end{equation}
    }
    \item {Modular Summation Gate $\SUM(b)$ where $b\in\intset{D}$:
    \begin{equation}
    \SUM(b)_h:\bra{a}_h\rightarrow \bra{a+b}_h.
    \end{equation}
    Note that $-b\equiv D-b\ (\mod D)$.
    }
    \item {Modular Multiplication Gate $\MUL(b)$ where $b\in\odd{D}$:
    \begin{equation}
    \MUL(b)_h:\bra{a}_h\rightarrow \bra{ab}_h.
    \end{equation}
    Note that $b$ is odd, so it is coprime with $D=2^d$ and then has a unique multiplicative inverse $b^{-1}\mod D$.
    }
    \item {
    Bi-particle Modular Summation Gate $\BSUM$:
    \begin{equation}
\BSUM_{(h,t)}:\bra{a}_h\bra{b}_t\rightarrow \bra{a}_h\bra{b+a}_t.
    \end{equation}
    }
    \item {Bitwise XOR Gate $\XOR$:
    \begin{equation}
\XOR_{(h,t)}:\bra{a}_h\bra{b}_t\rightarrow \bra{a}_h\bra{b\oplus a}_t.
    \end{equation}}
\end{enumerate}

\subsection{Fourier Entangled State}\label{sec2.3}
\begin{definition}[\textbf{Fourier Entangled (FE) State}]\label{def1} Let $X\in S_X$ be a secret, $C\in S_C$ be a random variable independent of $X$, $\forall x\in S_X, \forall C\in S_C, \Pr(X=x)=p_x=\frac{1}{\abs{S_X}}, \Pr(C=c)=p_c=\frac{1}{\abs{S_C}}$. Given two quantum systems $P,Q$ with $l_1,l_2$ qubits respectively, and one $D=2 ^ d$-dim FE state is as:
\begin{align}
&\bra{\psi(x,c)}_{(P,Q)}=\nonumber\\
&\ \ \sfrac{D}\sum_{j\in\intset{D}}\omega^{jh(x,c)}\bra{f(j,x,c)}_P\bra{g(j,x,c)}_Q,
\end{align}
where $f:\intset{D}\times S_X\times S_C\to \intset{2^{l_1}}$, $g:\intset{D}\times S_X\times S_C\to \intset{L}$, $L=2^{l_2}$. $\forall j',j \in \intset{D}, \innerpro{f(j',x,c)}{f(j,x,c)}=\delta_{j'j}$. 
\end{definition}

We call the systems $P$ and $Q$ local systems. We generally consider attacks on $Q$, not $P$. The function $h(x,c)$ is called Phase domain, correspondingly, the functions $f(j,x,c),g(j,x,c)$ are Bit domain. The information these functions contain are called Phase-information and Bit-information respectively.

Assume an FE state $\sfrac{D}\sum_{j\in\intset{D}}\bra{j}_P\bra{ja_1}_{Q_1}\bra{ja_2}_{Q_2}$, where $a_1,a_2\in\odd{D}$, and it has the following properties: 
\begin{property}[\textbf{Addition-Multiplication Independence}]\label{prop1} If using $\ROT(b_1)$ on $Q_1$, it will be attached a phase factor $\omega^{ja_1b_1}$. If using $\ROT(b_2)$ on $Q_2$, there is a new factor $\omega^{ja_2b_2}$. Therefore, the total phase is $\omega^{j(a_1b_1+a_2b_2)}$. 
\end{property}
\begin{property}[\textbf{Addend's Disappearance}]\label{prop2} Apply $\SUM(c)$ on $Q_1$, then apply $\ROT(b)$ on $\bra{ja_1+c}_{Q_1}$. Since the global phase factor $\omega^{bc}$ does not affect any measurement results \cite{2000NielsenQua}, it can be omitted, i.e., the addend $c$ does not affect the phase.
\end{property}

\section{Proposed Protocol}\label{sec3}
\noindent 
In this section, we first define the unconditional security under the malicious model, then briefly introduce Entanglement Bondage we will use in the protocol. We provide the specific protocol process at the end.

\subsection{Unconditional Security under the Malicious Model}\label{sec2.4}

\begin{definition}[\textbf{Malicious adversary model}]\label{def2} In this model, attacks other than a) forging input, b) not participating in the protocol, and c) terminating the protocol halfway should be all defensed or detected. 
\end{definition}
In addition, the following assumptions are also made to simplifies the analysis:
\begin{enumerate}[1)]
    \item Malicious adversaries do not input, but only steal, since the information it inputs will interfere with its attack.
    \item Malicious adversaries are to obtain information, not just destroy. We mainly ensure the privacy of valid information.
    \item If either party is malicious, the other is not, since a protocol executed between malicious participants is meaningless.
\end{enumerate}

Considering the particularity of quantum protocols, we use information theory language to define unconditional security under this model. Assume that in a two-party protocol $\Pi$, there is an honest party $HP$ and a malicious party $MP$ respectively. Denote the privacy of $HP$ is a random variable $X$, with a Shannon entropy $H(X)=m_X$. Denote the expected result $MP$ should get as $F$. Then
\begin{definition}[\textbf{Leakage Degree}]\label{def3} Under an attack $AT$, the view obtained by $MP$ is a random variable $Z$. Mutual information $I=H(Z:X)$ measures the information increment to $X$ when $Z$ is known. Stipulate that if attack $AT$ is detected by $HP$, $H(Z:X)=0$, since the attack has failed. If $MP$ cannot get $F$ after its attack, then the \textbf{leakage degree} of $X$ is defined as $H(Z:X)$; Otherwise, it equals to the \textbf{conditional mutual information}\cite{1978WynerDef} $H(Z:X|F)$, which measures how much information $MP$ will obtain other than $F$.
\end{definition}
\begin{definition}[\textbf{Negligibility}]\label{def4} A function $\mu(m_X):\mathbb{N}\to [0,1]$ is said to be \textbf{negligible}, if there is no positive polynomial $poly(m_X)$ about $m_X$ so that $\mu(m_X)=\Omega(1/poly(m_X))$.
\end{definition}
\begin{definition}[\textbf{Unconditional Security under the Malicious Model}]\label{def5} Protocol $\Pi$ is said to has unconditional security under the malicious model, if in one run of $\Pi$, the leakage degree of any party's privacy is negligible under all known malicious attacks.
\end{definition}

\subsection{Entanglement Bondage}\label{sec3.1}
\noindent 
Assume that in a quantum protocol, $MP$ should send several particles $t_1,t_2,g$ to $HP$, which are FE-entangled as $\sfrac{D}\sum_{j\in\intset{D}}\bra{j}_{t_1}\bra{ja_1}_{t_2}\bra{ja_2}_{g}$, $a_1,a_2\in\odd{D}$. However, $MP$ can perform a \textbf{forgery attack}, i.e., it doesn't entangle $t_1,t_2,g$. Now $HP$ wants to verify $MP$'s honesty, without disrupting the superposition state of the particles. $HP$ applies $\MUL,\BSUM$ to add the state $j,ja_1$ to $ja_2$: 
\begin{align}ja_2\to jk_1+jk_2a_1+jk_3a_2=j(k_1+k_2a_1+k_3a_2),\end{align}
where $k_1,k_2,k_3\in\odd{D}$. This process is called \textbf{Entanglement Bondage}. If $MP$ is dishonest, then this step will entangle the three particles. Under the entanglement, $MP$ cannot steal information without being detected.

After the above steps, $HP$ can perform an honesty test as follows. $HP$ sends $k_1,k_2,k_3$ to $MP$, and $MP$ will return an answer $r=(k_1+k_2a_1+k_3a_2)^{-1}$, since the add of any three odds is also an odd. $HP$ applies $\MUL(r)$ on $g$, then $bra{j(k_1+k_2a_1+k_3a_2)}_g\to\bra{j}_g$. Now $HP$ can perform $\XOR$ on $t_1,g$, then $\bra{j}_g\to\bra{0}_g$. $HP$ can verify the correctness of the state prepared by $MP$ by measuring $g$. If $g$ is in state $\bra{0}$, then $MP$ passes; Otherwise, $MP$'s cheating is detected. This process is actually a zero-knowledge proof, through which $HP$ can verify whether $MP$ really prepared the correct quantum state as promised, without measuring the state itself. The effectiveness of this mechanism can be seen in Section~\ref{sec4.2}.

\subsection{Specific Protocol Process}\label{sec3.2}
\noindent 

\begin{definition}[\textbf{Secure Two-party Scalar Product (S2SP)}]
Alice and Bob each have an $n$-dim vector $\mathbf{x}=(x_1,x_2,\cdots,x_n)$ and $\mathbf{y}=(y_1,y_2,\cdots,y_n)$ respectively, where $x_i,y_i\in \intset{N}$, $N=2^m$. Alice is to get 
$u=\mathbf{x}\cdot\mathbf{y} +v\mod N=\sum_{i=1}^{n}{x_iy_i} +v\mod N $, while $v\in\intset{N}$ is known only by Bob. A secure protocol should meet the following requirements:

$\bullet$ Alice's Privacy: Bob learns no information about $\mathbf{x}$.

$\bullet$ Bob's Privacy: Alice learns no information about $\mathbf{y},v$ other than $u=\mathbf{x}\cdot\mathbf{y}+v\mod N$.
\end{definition}

\begin{protocol}\textbf{Secure Two-party Quantum Scalar Product Protocol (S2QSPP).}\label{pro1}
\end{protocol}

\noindent \textit{1) Preparation Stage}
\begin{enumerate}[Step 1]
\item Alice and Bob set $d= m+2$ and $D=2^d$. 
\item Alice assigns $p_i=2x_i+1$ for each $i=\cint{n}$. 
\item Bob selects randomly $v_1,v_2,\cdots,v_{n-1}\in \intset{N}$, and then sets $v_n=\frac{4v-4\sum_{i=1}^{n-1}{v_i}\mod D}{4}$, i.e., $4\sum_{i=1}^{n}{v_i}\equiv 4v (\mod D)$. He assigns $q_i=2y_i+1$ and $s_i=4v_i-2y_i-1 \mod D$.
\end{enumerate}
\ \newline
\noindent \textit{2) Operation Stage}\newline
\noindent For each $i=\cint{n}$, do the following steps (all arithmetic operations are performed mod $D$ here):
\begin{enumerate}[Step 1]
    \item {(\textit{Alice's Inputing}) Alice prepares $4$ $d$-qubit particles $h,t_1,t_2,g$ initialized as $\bra{0}$. She selects $c_1,c_2,c_4\in\intset{D}$ and $c_3\in\odd{D}$ randomly, then does the following:
            \begin{align}
    &\bra{0}_{h}\bra{0}_{t_1}\bra{0}_{t_2}\bra{0}_{g}\nonumber\\
&\onarrow{\QFT_{h}}\sfrac{D}\sum_{j\in\intset{D}}\bra{j}_{h}\bra{0}_{t_1}\bra{0}_{t_2}\bra{0}_{g} \nonumber\\
&\onarrow{\XOR_{(h,t_1)}\XOR_{(h,t_2)}\XOR_{(h,g)}}\nonumber\\
&\sfrac{D}\sum_{j\in\intset{D}}\bra{j}_{h}\bra{j}_{t_1}\bra{j}_{t_2}\bra{j}_{g}\nonumber\\
&\onarrow{\MUL(p_i)_{t_2}\MUL(c_3)_{g}}\nonumber\\
&\sfrac{D}\sum_{j\in\intset{D}}\bra{j}_{h}\bra{j}_{t_1}\bra{jp_i}_{t_2}\bra{jc_3}_{g}\nonumber\\
&\onarrow{\SUM(c_1)_{t_1}\SUM(c_2)_{t_2}\SUM(c_4)_g}\sfrac{D}\sum_{j\in\intset{D}}\bra{j}_{h}\nonumber\\
&\bra{j+c_1}_{t_1}\bra{jp_i+c_2}_{t_2}\bra{jc_3+c_4}_{g}.
            \end{align}
    Now the particles $h,t_1,t_2,g$ are in an FE state. Alice then sends $t_1,t_2,g$ to Bob.}
\item {(\textit{Entanglement Bondage}) Bob selects $k_1,k_2,k_3\in\odd{D}$ randomly, then does the following:
            \begin{align*}
&\sfrac{D}\sum_{j\in\intset{D}}\bra{j}_{h}\bra{j+c_1}_{t_1}\bra{jp_i+c_2}_{t_2}\bra{jc_3+c_4}_{g}\\
&\onarrow{\MUL(k_1)_{t_1}\MUL(k_2)_{t_2}\MUL(k_3)_g}\\
&\sfrac{D}\sum_{j\in\intset{D}}\bra{j}_{h}\bra{jk_1+c_1k_1}_{t_1}\bra{jp_ik_2+c_2k_2}_{t_2}\\
\end{align*}
            \begin{align}
&\bra{jc_3k_3+c_4k_3}_{g}\nonumber\\
&\onarrow{\BSUM_{(t_1,g)}\BSUM_{(t_2,g)}}\nonumber\\
&\sfrac{D}\sum_{j\in\intset{D}}\bra{j}_{h}\bra{jk_1+c_1k_1}_{t_1}\bra{jp_ik_2+c_2k_2}_{t_2}\nonumber\\
&\bra{j\left(k_1+p_ik_2+c_3k_3\right)+\left(c_1k_1+c_2k_2+c_4k_3\right)}_{g}\nonumber\\
&=\sfrac{D}\sum_{j\in\intset{D}}\bra{j}_{h}\bra{jk_1+c_1k_1}_{t_1}\bra{jp_ik_2+c_2k_2}_{t_2}\nonumber\\
&\bra{jr_1+r_2}_{g}\nonumber\\
&\onarrow{\MUL(k_1^{-1})_{t_1}\MUL(k_2^{-1})_{t_2}}\sfrac{D}\sum_{j\in\intset{D}}\bra{j}_{h}\bra{j+c_1}_{t_1}\nonumber\\
&\bra{jp_i+c_2}_{t_2}\bra{jr_1+r_2}_{g},
      \end{align}
where $r_1=k_1+p_ik_2+c_3k_3$, $r_2=c_1k_1+c_2k_2+c_4k_3$.}

\item {(\textit{Bob's Inputing}) Bob does the following:
       \begin{align}
&\sfrac{D}\sum_{j\in\intset{D}}\bra{j}_{h}\bra{j+c_1}_{t_1}\bra{jp_i+c_2}_{t_2}\bra{jr_1+r_2}_{g}\nonumber\\
&\onarrow{\ROT(s_i)_{t_1}\ROT(q_i)_{t_2}}\sfrac{D}\omega^{(j+c_1)s_i}\omega^{(jp_i+c_2)q_i}\nonumber\\
&\sum_{j\in\intset{D}}\bra{j}_{h}\bra{j+c_1}_{t_1}\bra{jp_i+c_2}_{t_2}\bra{jr_1+r_2}_{g}\nonumber\\
&=\sfrac{D}\sum_{j\in\intset{D}}\omega^{jM_i}\bra{j}_{h}\bra{j+c_1}_{t_1}\bra{jp_i+c_2}_{t_2}\nonumber\\
&\bra{jr_1+r_2}_{g},
     \end{align}
where $M_i=s_i+p_iq_i$. We can omit the global phase $\omega^{c_1s_i+c_2q_i}$.}

\item {(\textit{Alice's Honesty Test}) Bob now verifies Alice's honesty as follows:
    \begin{enumerate}[a)]
    \item (\textit{Question}) Bob tells Alice the values of $k_1,k_2,k_3$.
    \item {(\textit{Answer}) Alice calculates $r_3=r_1^{-1}$ and $r_4=c_1-r_2r_3$, then tells Bob the values of $r_3,r_4$.}
    \item {(\textit{Verification}) Bob does the following:
        \begin{align}
&\sfrac{D}\sum_{j\in\intset{D}}\omega^{jM_i}\bra{j}_{h}\bra{j+c_1}_{t_1}\nonumber\\
&\bra{jp_i+c_2}_{t_2}\bra{jr_1+r_2}_{g}\nonumber\\
&\onarrow{\MUL(r_3)_g}\sfrac{D}\sum_{j\in\intset{D}}\omega^{jM_i}\bra{j}_{h}\bra{j+c_1}_{t_1}\nonumber\\
&\bra{jp_i+c_2}_{t_2}\bra{j+r_2r_3}_{g}\nonumber\\
&\onarrow{\SUM(r_4)_g}\sfrac{D}\sum_{j\in\intset{D}}\omega^{jM_i}\bra{j}_{h}\bra{j+c_1}_{t_1}\nonumber\\
&\bra{jp_i+c_2}_{t_2}\bra{j+c_1}_{g}\nonumber\\
&\onarrow{\XOR_{(t_1,g)}}\sfrac{D}\sum_{j\in\intset{D}}\omega^{jM_i}\bra{j}_{h}\bra{j+c_1}_{t_1}\nonumber\\
&\bra{jp_i+c_2}_{t_2}\bra{0}_{g},
      \end{align}
and measures $g$. Alice passes only if he gets $\bra{0}_g$.}
\end{enumerate}
If the test is passed, then Bob returns $t_1,t_2$ to Alice.}

\item {(\textit{Bob's Honesty Test A}) Alice now verifies Bob's honesty as follows:
       \begin{align}
&\sfrac{D}\sum_{j\in\intset{D}}\omega^{jM_i}\bra{j}_{h}\bra{j+c_1}_{t_1}\bra{jp_i+c_2}_{t_2}\nonumber\\
&\onarrow{\SUM(-c_1)_{t_1}\SUM(-c_2)_{t_2}}\nonumber\\
&\sfrac{D}\sum_{j\in\intset{D}}\omega^{jM_i}\bra{j}_{h}\bra{j}_{t_1}\bra{jp_i}_{t_2}\nonumber\\
&\onarrow{\MUL(p_i^{-1})_{t_2}}\sfrac{D}\sum_{j\in\intset{D}}\omega^{jM_i}\bra{j}_{h}\bra{j}_{t_1}\bra{j}_{t_2}\nonumber\\
&\onarrow{\XOR_{(h,t_1)}\XOR_{(h,t_2)}}\sfrac{D}\sum_{j\in\intset{D}}\omega^{jM_i}\bra{j}_{h}\nonumber\\
&\bra{0}_{t_1}\bra{0}_{t_2},
      \end{align} 
and measures $t_1,t_2$. Bob passes the test only if Alice gets $\bra{0}_{t_1}\bra{0}_{t_2}$.
}

\item {(\textit{Alice's Result}) Alice preforms $\QFIT$ on $h$:
\begin{equation}
\sfrac{D}\sum_{j\in\intset{D}}\omega^{jM_i}\bra{j}_{h}\onarrow{\QFIT_h}\bra{M_i}_{h}\mod D,
\end{equation}
and measures $h$ to obtain the result $M_i=s_i+p_iq_i$.
}

\end{enumerate}
The quantum circuit of all the steps is shown in figure~\ref{fig_total_latex}.

\begin{figure*}[!t]
\centering
\includegraphics[width=0.85\linewidth,height = 0.3\linewidth]{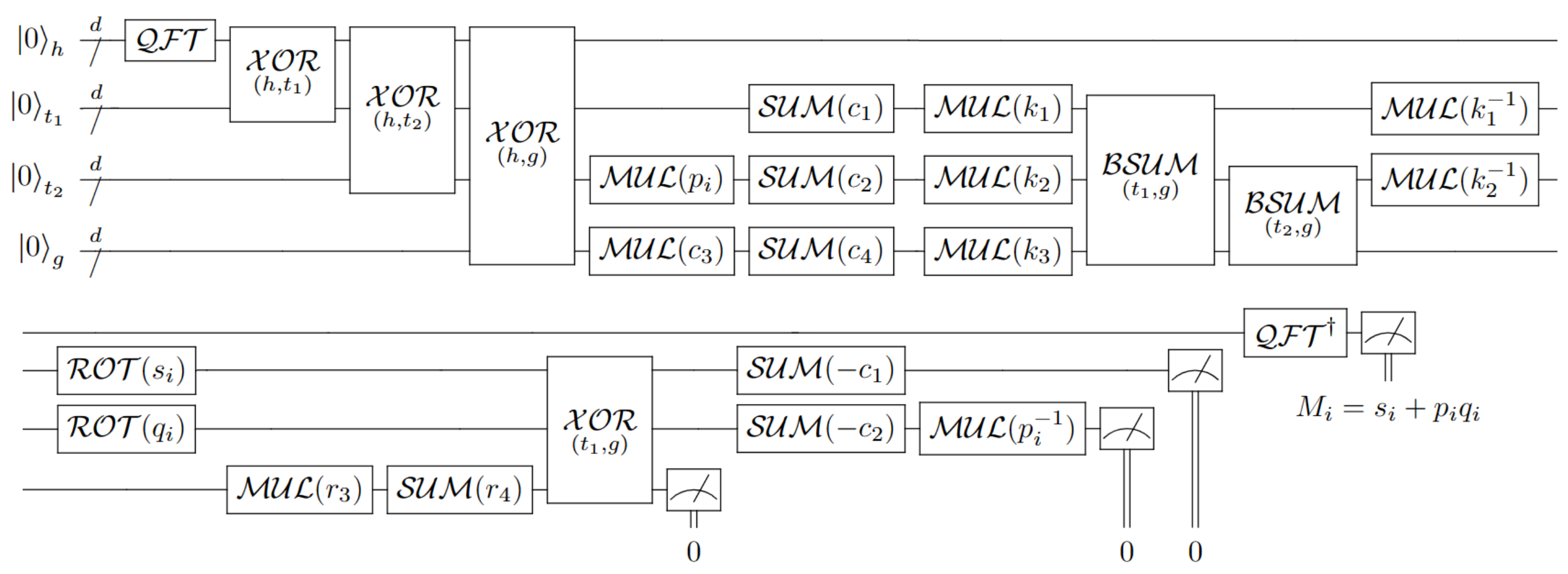}
\caption{The quantum circuit of Operation stage.}
\label{fig_total_latex}
\end{figure*}

\ 

\noindent \textit{3) Output Stage}\newline
\noindent Alice calculates as $Output=\frac{\sum^{n}_{i=1}\left(M_i-2x_i\right)\mod D}{4}$.

\section{Protocol Analysis}\label{sec4}
\subsection{Correctness}\label{sec4.1}
\begin{align*}
&Output=\frac{\sum^{n}_{i=1}\left(M_i-2x_i\right)\mod D}{4}\\
&=\frac{\sum^{n}_{i=1}\left(s_i+p_iq_i-2x_i\right)\mod D}{4}\\
&=\frac{\sum^{n}_{i=1}\left(4v_i+4x_iy_i\right)\mod D}{4}\\
\end{align*}
\begin{align}
&=\frac{4\sum^{n}_{i=1}v_i+4\sum^{n}_{i=1}x_iy_i\mod D}{4}\nonumber\\
&=\frac{4v+4\sum^{n}_{i=1}x_iy_i\mod D}{4}=\frac{4(v+\sum^{n}_{i=1}x_iy_i)-a2^d}{4}\nonumber\\
&=\frac{4(v+\sum^{n}_{i=1}x_iy_i-a2^m)}{4}=v+\sum^{n}_{i=1}x_iy_i-a2^m,
\end{align}
where $4(v+\sum^{n}_{i=1}x_iy_i)-a2^d\in \intset{D}$ and $d=m+2$, thus $v+\sum^{n}_{i=1}x_iy_i-a2^m \in \intset{\frac{D}{4}}=\intset{N}$. I.e.,
\begin{align}
&Output=v+\sum^{n}_{i=1}x_iy_i-a2^m\nonumber\\
&= v+\sum^{n}_{i=1}x_iy_i\mod D=u.
\end{align}

\subsection{Security}\label{sec4.2}
\noindent Since all the interactions occur in Operation stage, an attacker may try the following possible attacks.

\noindent\textit{1) Measurement Attack}

Measurement attack is to directly perform a local general measurement on the attacker's own particles. ``General quantum measurement'' is equivalent to a combination of introducing auxiliary systems, performing unitary operations and projective measurements\cite{2000NielsenQua}.

\noindent\textit{2) Entangle-measure Attack}

After receiving a particle $t$, the attacker can prepare an auxiliary particle $e$ and perform a unitary operation 
\begin{equation}
\U_{\varepsilon}:\bra{j}_t\bra{0}_e\rightarrow \sqrt{\eta_j}\bra{j}_t\bra{\varepsilon(j)}_e+\sqrt{1-\eta_j}\bra{V(j)}_{(t,e)},
\end{equation}
where $\eta$ is the probability that $\bra{j}_t$ is not changed. Then he sends $t$ back, and monitors its movement by measuring $e$. 

\noindent\textit{3) Forgery Attack}

Forgery attack is to send a particle in a forged rather than correct state to steal information, as a type of malicious attack.

\noindent\textit{4) Intercept-Resend Attack}

Similar to forgery attack, intercept-resend attack means to obtain information from a particle after receiving it, and then return a forged particle back.

\noindent\textit{5) False Verification Information Attack}

This attack may occur in Step 4 of Operation stage, which means to send fake verification information, such as $k_1,k_2,k_3$ and $r_3,r_4$.

\noindent\textit{6) External Attack}

External attack means that any eavesdropper Eve wants to steal Alice or Bob's privacy.

\noindent\textit{7) Semi-honest attack}

If the attacker is semi-honest, it may try to deduce information other than the expected result it should get, from all the values it can obtain during the protocol.

We has the following theorem.
\begin{theorem}\label{theorem1}
Protocol~\ref{pro1} holds unconditional security under the malicious model, i.e., it can resist the above attacks.
\end{theorem}

\begin{proof}

\noindent\textit{1) Measurement Attack}

First, we give the following Lemma~\ref{lemma1}. It is proved in Appendix~\ref{secA.1}, relying on the use of partial trace $tr_P(\cdot)$.
\begin{lemma}[Security of Phase-information]\label{lemma1}
Under the entanglement of an FE state as Definition~\ref{def1}, any attacker who only owns local system $Q$ cannot extract the phase-information.
\end{lemma}
$\bullet$ \textbf{Alice's Privacy:} Denote $X_A=x_i$. To steal the largest information, Bob can wait for $r_3,r_4$, then measure $t_1,t_2,g$. For the leakage degree of Alice's privacy $I_A=H(Z_B:X_A)$, we have the following Lemma~\ref{lemma2}. Its proof relies on the Holevo bound\cite{1973HolevoSta}, see Appendix~\ref{secA.2} for details.
\begin{lemma}[Security of Bit-information]\label{lemma2}
Under the measurement attack, $I_A=0$, i.e., Bob cannot obtain any valid information about $X_A$. Besides, this effect cannot be achieved without particle $g$.
\end{lemma}

$\bullet$ \textbf{Bob's Privacy:} Denote $X_B=(y_i,v_i)$. Until Step 5, Bob owns a local system $(t_1,t_2,g)$ of the FE state. By Lemma~\ref{lemma1}, Alice cannot obtain any information about $q_i,s_i$, because they are all on the phase. I.e., The leakage degree of Bob's privacy $I_A=H(Z_A:X_B)=0$.\newline

\noindent\textit{2) Entangle-measure Attack}

$\bullet$ \textbf{Alice's Privacy:} Denote $X_A=x_i$. After Bob resent, it is possible to obtain information by measuring only in Step 6, since only now Alice carries on operations. Because Alice performs an honest test on the sent particles $t_1,t_2$ in Step 5, if they are no longer in their original states, the attack will be detected. Therefore, $\eta_j$ should be set to $1$, i.e., $\U_{\varepsilon}:\bra{j}_t\bra{0}_e\rightarrow \bra{j}_t\bra{\varepsilon(j)}_e$. 
    
Assume that Bob owns all particles $t_1,t_2,g,e$ to steal the largest information, and performs $\U_{\varepsilon}$. Note that the Holevo bound of Bob's particle now equals to that in measurement attack, i.e., $H(Z_B:X_A)=0$ (see the proof of Lemma~\ref{lemma2}), because $\U_{\varepsilon}$ is a local quantum operation and does not increase the Holevo bound (Chapter 12 Problem 12.1 of Ref.~\cite{2000NielsenQua}). If he returns any particle, he cannot steal more, since discarding quantum systems {is also a kind of quantum operations and does not increase the Holevo bound}. Alice's any operations in Step 6 are equivalent to local quantum measurement, which doesn't affect Bob's measurement results by the principle of implicit measurement (Chapter 4.4 of Ref.~\cite{2000NielsenQua}). Therefore, it also does not increase the Holevo bound of Bob's particles. In this way, we can immediately deduce that $I_A=0$, consistent with the measurement attack.

$\bullet$ \textbf{Bob's Privacy:} Alice cannot perform entangle-measure attacks, since she may only send particles once.\newline

\noindent\textit{3) Forgery Attack}

$\bullet$ \textbf{Alice's Privacy:} Since Alice won't return the particles sent by Bob, Bob cannot perform forgery attacks.

$\bullet$ \textbf{Bob's Privacy:} {Denote $X_B=(y_i,v_i)$.} Since Bob's information only exists on the phase domain, the only known way to obtain it is to use $\QFIT$. By Lemma~\ref{lemma1}, it is impossible to obtain phase-information under entanglement, then particles $t_1,t_2$ must be non-entangled. Accordingly, we have the following Lemma~\ref{lemma3}. We prove it by using the Holevo bound\cite{1973HolevoSta} in Appendix~\ref{secA.3}. 
\begin{lemma}[Security under Forgery Attack]\label{lemma3}
Due to the entanglement bondage in Step 2 and the test in Step 4, the leakage degree of Bob's privacy under the forgery attack is $I_B<\frac{d^2+d-2}{2\cdot2^d}=O\left( \frac{d^2}{2^d}\right)$, with an asymptotic coefficient $\frac{1}{2}$.
\end{lemma}

$I_B$ is negligible, since $2^d$ is exponential level. When $d$ reaches its minimum value, i.e., $d=1+2=3$, then $\frac{d^2+d-2}{2\cdot2^d}=\frac{5}{8}$. This bound will soon approach 0 as $d$ increases.\newline

\noindent\textit{4) Intercept-Resend Attack}

$\bullet$ \textbf{Alice's Privacy:} If Bob sends any forged particle back in Step 4, then as in the measurement attack, the leakage degree of Alice's privacy $I_A=0$.

$\bullet$ \textbf{Bob's Privacy:} Similar to forgery attacks, it is impossible for Alice to perform this attack.\newline

\noindent\textit{5) False Verification Information Attack}

$\bullet$ \textbf{Alice's Privacy:} If Bob did not send the correct $k_1,k_2,k_3$, we have $I_A=0$, just like the measurement attack.

$\bullet$ \textbf{Bob's Privacy:} Similarly, if Alice sends incorrect $r_3,r_4$, then $I_B=0$, and Bob can detect it in Step 4 c).\newline

\noindent\textit{6) External Attack}

According to the analysis above, if Eve intercepts the particles sent by Alice to Bob, obviously she will get nothing. Similarly, if she intercepts the particles returned by Bob to Alice, she cannot get Bob's privacy.\newline

\noindent\textit{7) Semi-honest Attack}

$\bullet$ \textbf{Alice's Privacy:} If Bob is semi-honest, then all he can learn are $r_3,r_4$. He won't obtain any information, as well as in the measurement attack.

$\bullet$ \textbf{Bob's Privacy:} {Denote $X_B=(\mathbf{y},v)$.} If Alice is semi-honest, all she can learn are $M_i\equiv p_i q_i+s_i\equiv 4x_i y_i+4v_i(\mod D)$. She may try to learn any information other than $u=\sum_{i=1}^n{x_i y_i}+v \mod N$. We have the following Lemma~\ref{lemma4}, and prove it in Appendix~\ref{secA.4}, by directly calculate the Shannon entropy. 
\begin{lemma}[Security under Semi-honest Attack]\label{lemma4}
Even if Alice knows all $M_i$, the leakage degree of Bob's privacy $I_B=0$.
\end{lemma}

In total, Protocol~\ref{pro1} has unconditional security under the malicious model.\end{proof}

\subsection{Performance}\label{sec4.3}
\noindent We take the basic 1-, 2- and 3-qubit quantum gates as the measurement unit of computational complexity, such as Hadamard Gate $\H:\bra{a}\rightarrow \frac{\bra{0}+(-1)^a\bra{1}}{\sqrt{2}}$, Controlled-X Gate $\CNOT:\bra{a}\bra{b}\rightarrow \bra{a}\bra{b\oplus a}$, Z-axis Rotation Gate $\PROT(i):\bra{a}\rightarrow e^{\frac{\imath 2\pi 2^i}{D}a}\bra{a}$ ($i\in\intset{d}$), it's controlled version $\CPROT(i):\bra{a}\bra{b}\rightarrow e^{\frac{\imath 2\pi 2^i}{D}ab}\bra{a}\bra{b}$, and Toffoli Gate $\Toffoli:\bra{a}\bra{b}\bra{c}\rightarrow \bra{a}\bra{b}\bra{c\oplus a\cdot b}$, etc. 

In general, the complexity of the gates in Section~\ref{sec2.2} are all below $O(d^2)$ (see Appendix~\ref{secB.1} for details). Since $d=O(m)$, the total complexity of Protocol~\ref{pro1} is $O(nm^2)$. Since there are only 4 $d$-qubit particles are sent for each $i=1,2,\cdots,n$, the communication complexity is $O(nm)$.

We use $m,n$ to represent the bit number and dimension of input vectors respectively. See Table~\ref{table_com} for the comparison between our protocol and the previous. In Ref.~\cite{2012HeA}, $O\left(4^mn\log^2 n\right)$ entanglements should have been prepared and sent. In Ref.~\cite{2018WangQua}, for real vectors $\mathbf{x}$, $\mathbf{y}$, $\innerpro{\mathbf{x}}{\mathbf{y}}$ was evaluated with accuracy $\epsilon$. Its computational and communication complexity are $O(n\epsilon^{-2})$ and $O(2\epsilon^{-2}+n^2)$ respectively. Let $\abs{\mathbf{x}}$, $\abs{\mathbf{y}}=\Theta \left(2^m\right)$, then $\epsilon=\Theta(2^{-2m})$ is needed for the error of $\mathbf{x}\cdot\mathbf{y}$ to be less than $1$. Grover's algorithm was used in Ref.~\cite{2019ShiStr}, with complexity $O\left(\sqrt{2^m}\right)$. It can be seen that our protocol is polynomial in terms of computational and communication complexity, while the previous protocols have at least one complexity close to exponential. In addition, our protocol does not require a third party. The above proves its advantages.

\begin{table}[h]
\begin{center}
\caption{Performance Comparison}\label{table_com}
\begin{tabular}{|c|c|c|c|c|}
\hline
\multirow{2}*{Protocols} & Third & Quantum & Computational & Communication\\
~ & Party & Resource & Complexity & Complexity\\
\hline
He\cite{2012HeA}   & $\checkmark$ &  Qubit & $O\left(4^mn\log^{2} n\right)$&$O\left(4^mn\log^{2} n\right)$\\
\hline
Wang\cite{2018WangQua}   & $\times$ &  Qumode & $O\left(n2^{4m}\right)$ & $O\left(2^{4m}\right)$\\
\hline
Shi\cite{2019ShiStr}   & $\times$ &  Qubit & $O\left(\sqrt{2^m}\right)$ & $O\left(m\right)$\\
\hline
Our   & $\times$ &  Qubit & $O\left(nm^2\right)$ & $O\left(nm\right)$\\
\hline
\end{tabular}
\end{center}
\end{table}

\subsection{Experiment}\label{sec4.4}
\noindent We verify the correctness and the feasibility of our protocol by circuit simulation experiments in IBM Qiskit simulator (Qiskit-0.41.0; Python-3.7; OS-Windows). Without loss of generality, let's set $m=2$ (i.e., $d=4$). The circuits of all quantum gates we used are described in figure~\ref{fig_circuit}. Because the complexity of classical simulation is sensitive to qubits' number, we use Draper's adder\cite{2000DraperAdd}, as shown in figure~\ref{fig_SUM} and \ref{fig_BSUM}. It requires no auxiliary qubits, but has higher complexity $O(d^2)$. Besides, we design a special circuit for module multiplication on $\intset{D}$ as shown in figure~\ref{fig_MUL} to further reduce qubits. See Appendix~\ref{secB.2} for details of this design. Finally, we omit the measurement of particles $t_1,t_2,g$ and only focus on the output results on particle $h$. The total circuit of Protocol~\ref{pro1} is shown in figure~\ref{fig_total}. We execute the experiment two times. Table~\ref{table_sim_1} shows the first input and output, with the selection of the intermediate parameters $v_i, c_1,c_2,c_3,c_4,k_1,k_2,k_3$. Similarly, table~\ref{table_sim_2} describe the second experiment. Each quantum program for $i=1,2,3,4$ is executed 1000 times, and figure~\ref{fig_sim} shows the results. It can be seen that our protocol can be run successfully with 100\% probability, so it is correct and feasible.

\begin{figure*}[!t]
\captionsetup[subfigure]{margin=120pt} 
\subfloat[]{\label{fig_QFT}
\includegraphics[width=0.39\linewidth,height = 0.13\linewidth]{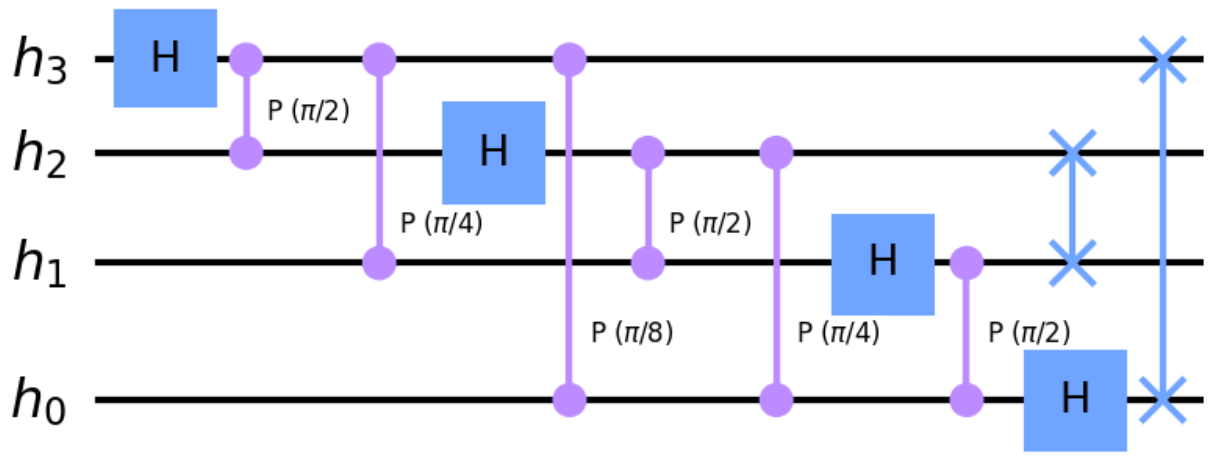}}
\hfill
\subfloat[]{\label{fig_ROT}
\includegraphics[width=0.08\linewidth,height = 0.13\linewidth]{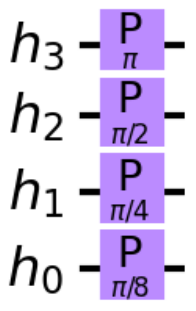}}
\hfill
\subfloat[]{\label{fig_SUM}
\includegraphics[width=0.23\linewidth,height = 0.13\linewidth]{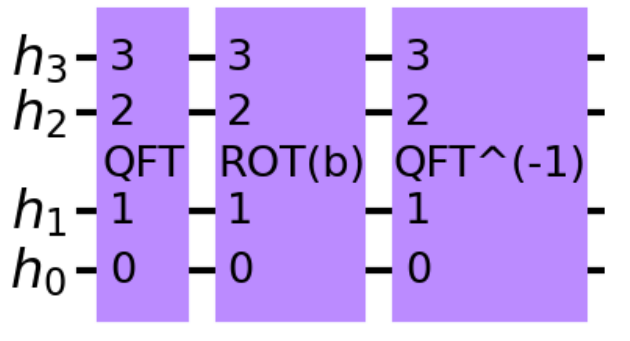}}
\hfill
\subfloat[]{\label{fig_MUL}
\includegraphics[width=0.45\linewidth,height = 0.13\linewidth]{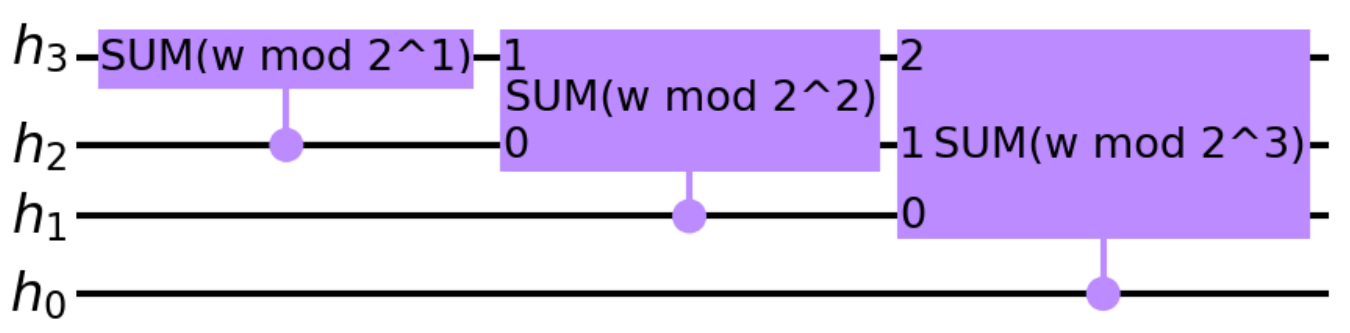}}
\hfill
\subfloat[]{\label{fig_BSUM}
\includegraphics[width=0.33\linewidth,height = 0.13\linewidth]{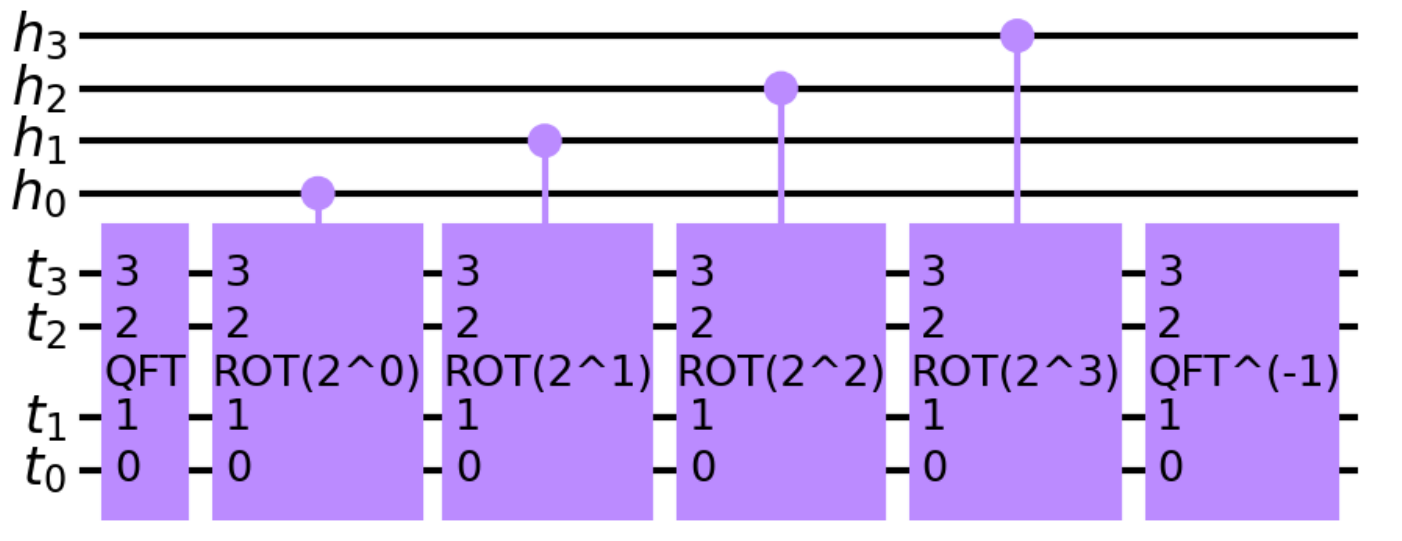}}
\hfill
\subfloat[]{\label{fig_XOR}
\includegraphics[width=0.14\linewidth,height = 0.13\linewidth]{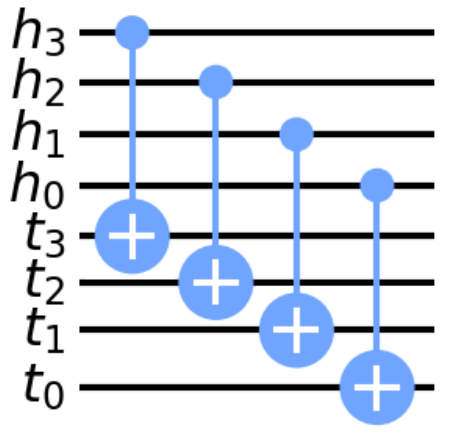}}
\caption{The circuits of quantum gates. a), b), c), d), e), f) are circuits of $\QFT$, $\ROT$, $\SUM$, $\MUL$, $\BSUM$, $\XOR$, respectively.}
\label{fig_circuit}
\end{figure*}

\begin{figure*}[!t]
\centering
\includegraphics[width=1.0\linewidth,height = 0.28\linewidth]{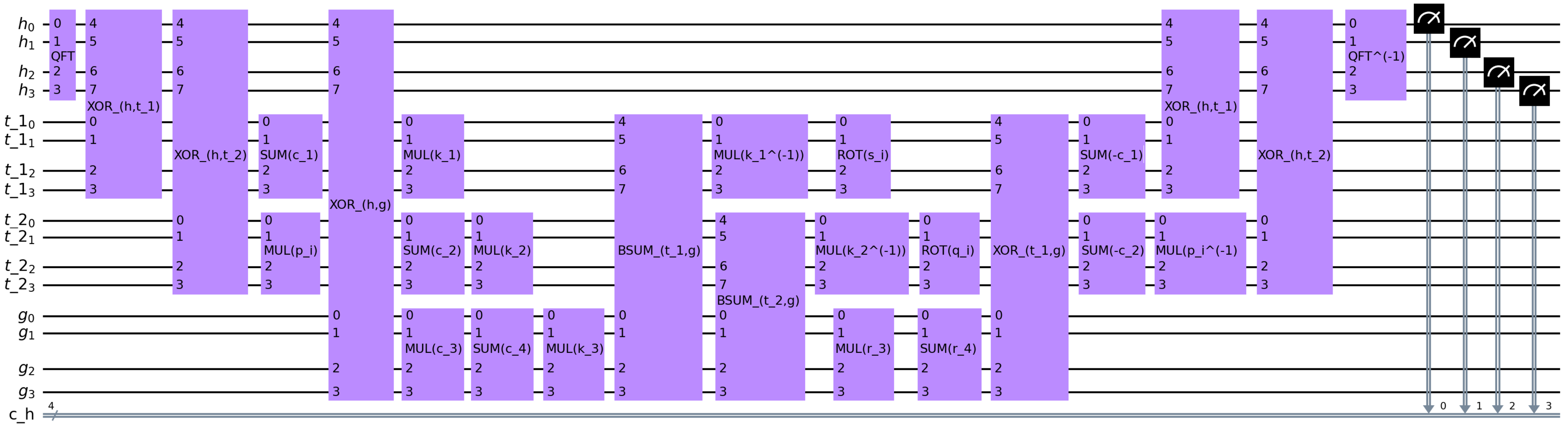}
\caption{The total circuit of our protocol.}
\label{fig_total}
\end{figure*}

\begin{table*}[h]
\begin{center}
\caption{Parameter of the First Experiment}\label{table_sim_1}
\begin{tabular}{|c|c|c|c|c|c|c|c|c|c|c|c|c|c|c|c|c|c|c|c|c|c|c|c|c|}
\hline
\multicolumn{2}{|c|}{External Parameter} & \multicolumn{4}{c|}{Input}& \multicolumn{17}{c|}{Intermediate Parameter} & \multicolumn{2}{c|}{Output}\\
\hline
$m$&$N$&$i$& $x_i$ & $y_i$ &$v$ &$d$&$D$ &$v_i$&$p_i$&$q_i$&$s_i$ &$c_1$&$c_2$&$c_3$&$c_4$&$k_1$&$k_2$&$k_3$&$r_1$ &$r_2$& $r_3$ & $r_4$& $M_i$ & $u$ \\
\hline
\multirow{4}{*}{2}&\multirow{4}{*}{4} &1 & 1 & 0 & \multirow{4}{*}{1} &\multirow{4}{*}{4}& \multirow{4}{*}{16}& 0 & 3 & 1 & 15 & 2 & 15 & 9 & 12 & 5 & 9 & 5 & 13 & 13 & 5 & 1 & 2 & \multirow{4}{*}{0}\\
\cline{3-5} \cline{9-24}
~& ~& 2 & 0 & 3 & ~& ~ & ~ & 11 & 1 & 7 & 5 & 10 & 6 & 11 & 3 & 15 & 5 & 7 & 1 & 9 & 1 & 1 & 12 & ~\\
\cline{3-5} \cline{9-24}
~& ~& 3 & 1 & 1 & ~& ~ & ~ & 15 & 3 & 3 & 9 & 10 & 11 & 3 & 2 & 1 & 1 & 1 & 7 & 7 & 7 & 9 & 2 & ~\\
\cline{3-5} \cline{9-24}
~& ~& 4 & 2 & 3 & ~& ~ & ~ & 13 & 5 & 7 & 5 & 4 & 9 & 1 & 6 & 7 & 9 & 15 & 3 & 7 & 11 & 7 & 8 & ~\\
\hline
\end{tabular}
\end{center}
\end{table*}

\begin{table*}[h]
\begin{center}
\caption{Parameter of the Second Experiment}\label{table_sim_2}
\begin{tabular}{|c|c|c|c|c|c|c|c|c|c|c|c|c|c|c|c|c|c|c|c|c|c|c|c|c|}
\hline
\multicolumn{2}{|c|}{External Parameter} & \multicolumn{4}{c|}{Input}& \multicolumn{17}{c|}{Intermediate Parameter} & \multicolumn{2}{c|}{Output}\\
\hline
$m$&$N$&$i$& $x_i$ & $y_i$ &$v$ &$d$&$D$ &$v_i$&$p_i$&$q_i$&$s_i$ &$c_1$&$c_2$&$c_3$&$c_4$&$k_1$&$k_2$&$k_3$&$r_1$ &$r_2$& $r_3$ & $r_4$& $M_i$ & $u$ \\
\hline
\multirow{4}{*}{2}&\multirow{4}{*}{4} &1 & 2 & 3 & \multirow{4}{*}{3} &\multirow{4}{*}{4}& \multirow{4}{*}{16}& 14 & 5 & 7 & 1 & 7 & 3 & 5 & 1 & 1 & 7 & 5 & 13 & 1 & 5 & 2 & 4 & \multirow{4}{*}{2}\\
\cline{3-5} \cline{9-24}
~& ~& 2 & 3 & 1 & ~& ~ & ~ & 10 & 7 & 3 & 5 & 1 & 14 & 5 & 11 & 5 & 11 & 7 & 5 & 12 & 13 & 5 & 10 & ~\\
\cline{3-5} \cline{9-24}
~& ~& 3 & 1 & 2 & ~& ~ & ~ & 15 & 3 & 5 & 7 & 14 & 12 & 1 & 9 & 1 & 3 & 3 & 13 & 13 & 5 & 13 & 6 & ~\\
\cline{3-5} \cline{9-24}
~& ~& 4 & 0 & 1 & ~& ~ & ~ & 0 & 1 & 3 & 13 & 4 & 1 & 1 & 1 & 13 & 3 & 15 & 15 & 6 & 15 & 10 & 0 & ~\\
\hline
\end{tabular}
\end{center}
\end{table*}

\begin{figure*}[!t]
\captionsetup[subfigure]{margin=120pt} 
\subfloat[]{\label{fig_sim_1_1}
\includegraphics[width=0.11\linewidth,height = 0.14\linewidth]{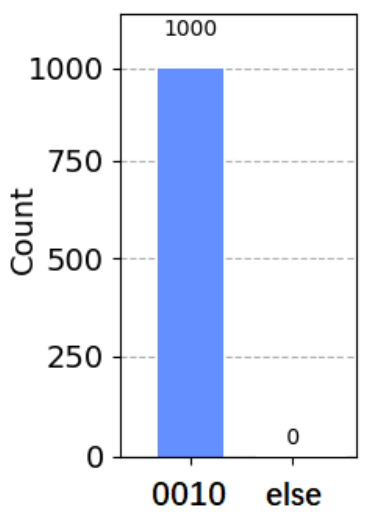}}
\hfill
\subfloat[]{\label{fig_sim_1_2}
\includegraphics[width=0.11\linewidth,height = 0.14\linewidth]{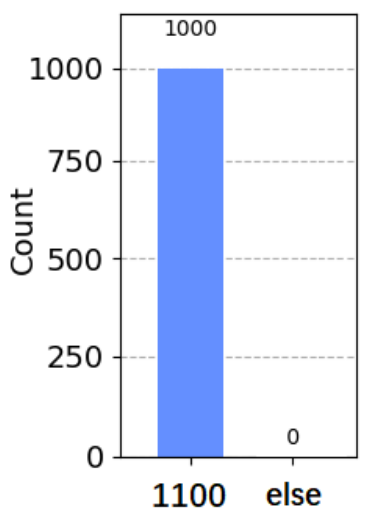}}
\hfill
\subfloat[]{\label{fig_sim_1_3}
\includegraphics[width=0.11\linewidth,height = 0.14\linewidth]{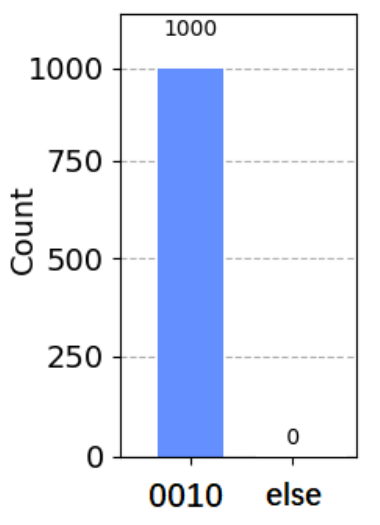}}
\hfill
\subfloat[]{\label{fig_sim_1_4}
\includegraphics[width=0.11\linewidth,height = 0.14\linewidth]{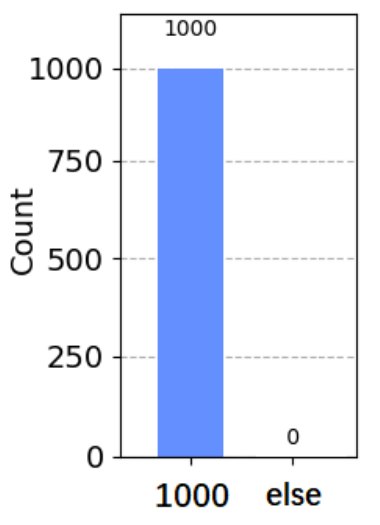}}
\hfill
\subfloat[]{\label{fig_sim_2_1}
\includegraphics[width=0.11\linewidth,height = 0.14\linewidth]{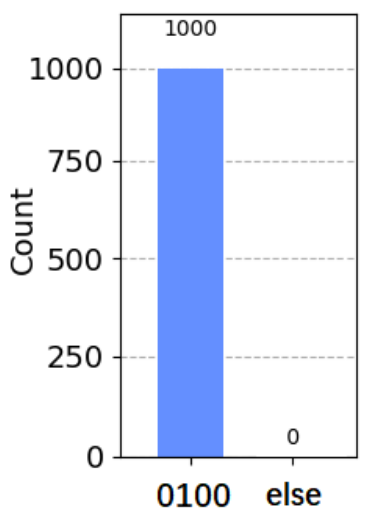}}
\hfill
\subfloat[]{\label{fig_sim_2_2}
\includegraphics[width=0.11\linewidth,height = 0.14\linewidth]{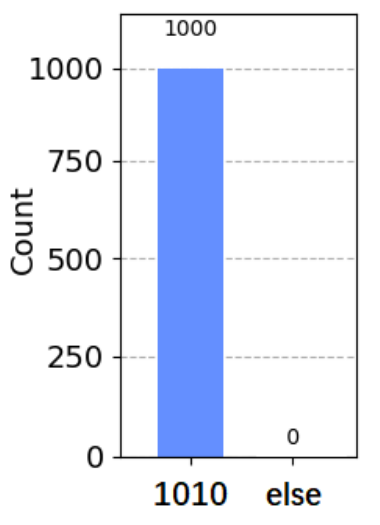}}
\hfill
\subfloat[]{\label{fig_sim_2_3}
\includegraphics[width=0.11\linewidth,height = 0.14\linewidth]{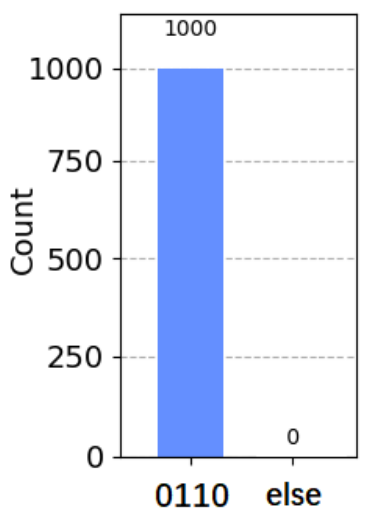}}
\hfill
\subfloat[]{\label{fig_sim_2_4}
\includegraphics[width=0.11\linewidth,height = 0.14\linewidth]{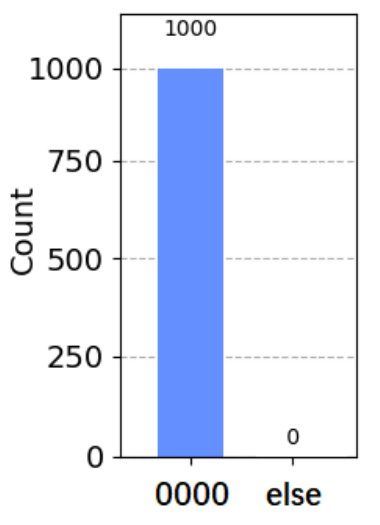}}
\caption{a), b), c), d) are the results $M_i$ for $i=1,2,3,4$ of the first experiment, respectively. Similarly, e), f), g), h) are the results of the second experiment.}
\label{fig_sim}
\end{figure*}

\section{Application}\label{sec5}
\noindent In this section, we present an application of Protocol~\ref{pro1}, i.e., a Privacy-preserving Two-party Matrix Multiplication (P2MM) protocol. This problem has been extensively studied in classical SMC. The complex matrix computation is mainly realized by generating product triples, among which cryptographic techniques such as oblivious transfer\cite{2016KellerMas,2018CramerSpd}, homomorphic encryption\cite{2018KellerOve,2019RatheeImp,2020ChenMal} and so on are widely used. We point out that since we have implemented a highly efficient two-party scalar product protocol, we no longer need to perform these computationally expensive processes. To our knowledge, this is the first quantum solution to solve this problem.

\subsection{Proposed P2MM Protocol}\label{sec5.1}

Firstly, we provide a precise definition of the problem.
\begin{definition}[\textbf{Privacy-preserving Two-party Matrix Multiplication (P2MM)}]
Alice and Bob have two $k\times n$ matrix 
\begin{equation}
\begin{aligned}
    \mathbf{A}=\left(a_{ij}\right)_{k\times n}=\begin{pmatrix} a_{11}& a_{12} & \cdots & a_{1n}\\a_{21} & a_{22} & \cdots & \vdots\\ \vdots & \vdots & \ddots & \vdots\\ a_{k1} & a_{k2} & \cdots & a_{kn}\end{pmatrix}
    \end{aligned}
\end{equation}
and $\mathbf{B}=\left(b_{ij}\right)_{k\times n}$, respectively, where $a_{ij},b_{ij}\in \intset{N}$, $N=2^m$. Alice is to get $\mathbf{U}=\mathbf{A}\cdot\mathbf{B} +\mathbf{V}$, while $\mathbf{V}=\left(v_{ij}\right)_{k\times n}, v_{ij}\in\intset{N}$ is a random matrix known only by Bob. Neither Alice nor Bob can get more information. Here we  omit ``$\mod N$''.
\end{definition}

The main scheme for calculating the matrix product of two participants is as follows: By using the formula of matrix multiplication, the calculation of $k\times n$ matrices is transformed into $kn$ times vector scalar product process, which are solved using our proposed S2QSP protocol.

\begin{protocol}\textbf{Privacy-preserving Two-party Quantum Matrix Multiplication Protocol (P2QMMP).}\label{pro2}
\end{protocol}

\noindent For each $1\le i\le k,1\le j\le n$, Alice and Bob do the following steps:
\begin{enumerate}[Step 1]
    \item {Alice separately extracts the $i$-th $n$-dimensional row vector of matrix $\mathbf{A}$, i.e.,
    \begin{align}\mathbf{x}_{i}=\left(a_{i1},a_{i2},\cdots,a_{in}\right),\end{align} 
    as her input vector.}
    \item {Similarly, Bob separately extracts the $j$-th $n$-dimensional column vector of matrix $\mathbf{B}$, i.e.,
    \begin{align}\mathbf{y}_{j}=\left(b_{1j},b_{2j},\cdots,b_{nj}\right).\end{align}
    Then he takes out the element $v_{ij}$ in the $i$-th row and $j$-column of matrix $\mathbf{V}$. The vector $\mathbf{y}_j$ and integer $v_{ij}$ are his inputs.}
    \item {Now the two vectors $\mathbf{x}_{i}$, $\mathbf{y}_{j}$, and the random integer $v_{ij}$ are valid inputs of Protocol~\ref{pro1}. Alice and Bob execute Protocol~\ref{pro1}, where parameter $N$, $m$ and $n$ are all set to the same as here.}
    \item {After all the steps of Protocol~\ref{pro1} are completed, Alice can obtain a corresponding result
    \begin{align}Output_{ij}=\mathbf{x}_{i}\cdot \mathbf{y}_{j}+v_{ij}.\end{align}
    }
\end{enumerate}
\noindent After executing $m\times n$ times the above steps, Alice now has $m\times n$ integers $Output_{ij}$, for $1\le i\le k,1\le j\le n$. She now assembles the result matrix by these integers as
\begin{align}Output=(Output_{ij})_{k\times n}.\end{align}

\subsection{Protocol Analysis}\label{sec5.2}

\noindent \textit{1) Correctness}
    \begin{align*}
Output=\begin{pmatrix}\mathbf{x}_{1}\cdot \mathbf{y}_{1}+v_{11} & \cdots & \mathbf{x}_{1}\cdot \mathbf{y}_{n}+v_{1n}\\\mathbf{x}_{2}\cdot \mathbf{y}_{1}+v_{21} & \cdots & \vdots\\ \vdots &  \ddots & \vdots\\ \mathbf{x}_{k}\cdot \mathbf{y}_{1}+v_{k1}  & \cdots & \mathbf{x}_{k}\cdot \mathbf{y}_{n}+v_{kn}\end{pmatrix}
\end{align*}
\begin{align*}
    &=\begin{pmatrix} \mathbf{x}_{1}\cdot \mathbf{y}_{1}& \mathbf{x}_{1}\cdot \mathbf{y}_{2} & \cdots & \mathbf{x}_{1}\cdot \mathbf{y}_{n}\\\mathbf{x}_{2}\cdot \mathbf{y}_{1} & \mathbf{x}_{2}\cdot \mathbf{y}_{2} & \cdots & \vdots\\ \vdots & \vdots & \ddots & \vdots\\ \mathbf{x}_{k}\cdot \mathbf{y}_{1} & \mathbf{x}_{k}\cdot \mathbf{y}_{2} & \cdots & \mathbf{x}_{k}\cdot \mathbf{y}_{n}\end{pmatrix}+\mathbf{V}
    \end{align*}
\begin{align*}
    &=\begin{pmatrix} \sum_{i=1}^{n}{a_{1i}b_{i1}}& \sum_{i=1}^{n}{a_{1i}b_{i2}} & \cdots & \sum_{i=1}^{n}{a_{1i}b_{in}}\\\sum_{i=1}^{n}{a_{2i}b_{i1}} & \sum_{i=1}^{n}{a_{2i}b_{i2}} & \cdots & \vdots\\ \vdots & \vdots & \ddots & \vdots\\ \sum_{i=1}^{n}{a_{ki}b_{i1}} & \sum_{i=1}^{n}{a_{ki}b_{i2}} & \cdots & \sum_{i=1}^{n}{a_{ki}b_{in}}\end{pmatrix}+\mathbf{V}\\
    \end{align*}
\begin{align}
    &=\begin{pmatrix} a_{11}& a_{12} & \cdots & a_{1n}\\a_{21} & a_{22} & \cdots & \vdots\\ \vdots & \vdots & \ddots & \vdots\\ a_{k1} & a_{k2} & \cdots & a_{kn}\end{pmatrix}\cdot\begin{pmatrix} b_{11}& b_{12} & \cdots & b_{1n}\\b_{21} & b_{22} & \cdots & \vdots\\ \vdots & \vdots & \ddots & \vdots\\ b_{k1} & b_{k2} & \cdots & b_{kn}\end{pmatrix}\nonumber\\
    &+\mathbf{V}=\mathbf{A}\cdot \mathbf{B}+\mathbf{V}.
    \end{align}

\noindent \textit{2) Security}
\begin{theorem}\label{theorem8}
Protocol~\ref{pro2} has unconditional security under the malicious model.
\end{theorem}
\begin{proof}
For each $1\le i\le k$, $1\le j \le n$, Alice and Bob execute one time Protocol~\ref{pro1}, and then Alice gets $Output_{ij}=\mathbf{x}_i\cdot \mathbf{y}_j+v_{ij}$, Bob only gets $v_{ij}$. Since Protocol~\ref{pro1} is secure enough with high probability (see Section~\ref{sec4.2}), the result Alice can only obtain is $Output_{ij}$. Since $Output_{ij}$ is part of the result matrix $\mathbf{U}$, i.e. it can be regressed from $\mathbf{U}$, it does not provide any information beyond what Alice deserves. Similarly, Bob is unable to receive any improper benefits. Therefore, Protocol~\ref{pro2} is at least as secure as Protocol~\ref{pro1}.
\end{proof}\newline

\noindent \textit{3) Performance}

\noindent Obviously, the computational and communication complexity of Protocol~\ref{pro2} are $O(kn\cdot nm^2)=O(kn^2m^2)$ and $O(kn\cdot nm)=O(kmn^2)$, since Protocol~\ref{pro1} is executed $kn$ times.

\section{Conclusion}\label{sec6}
\noindent In this paper, we propose a secure and efficient two-party quantum scalar product protocol, where several special properties of Fourier entangled states are used for calculation and security. Our protocol does not require any third parties, and has unconditional security under the malicious adversary model. It has polynomial level computational and communication complexity, which is the most efficient than the state-of-the-art protocols. Furthermore, based on the proposed S2QSP protocol, we present
a privacy-preserving matrix multiplication protocol as its extended application.

However, because our protocol involves high-dimensional entangled states, it will be relatively fragile under noise. The transmission error may be reduced by high-dimensional error correction code. Besides, there is a future research direction on how to extend the protocol to multi-party scenario, which can achieve a wider application.

\section*{Declarations}

\begin{itemize}
\item \textbf{Conflict of interest} The authors declare that they have no conflict of interest.
\item \textbf{Ethical statement} Articles do not rely on clinical trials. 
\item \textbf{Data availability} Data sharing does not applicable to this article as no datasets were generated or analysed during the current study.
\end{itemize}

\appendices
\section{Proof of Lemmas}\label{secA}

\subsection{Proof of Lemma~\ref{lemma1}}\label{secA.1}
\begin{proof}[]Let's assume that there is an attacker who dose not know the values of $X,C$. He may perform any type of measurement on local system $Q$ to obtain a measure result $Z$. First, we calculate the global density operator as
\begin{align}
&\rho^{(P,Q)}_{xc}=\outerpro{\psi(x,c)}{\psi(x,c)}_{(P,Q)}\nonumber\\
&=\frac{1}{D}\sum_{j',j\in \intset{D}}\omega^{(j'-j)h(x,c)}\outerpro{f(j',x,c)}{f(j,x,c)}_A\nonumber\\
&\outerpro{g(j',x,c)}{g(j,x,c)}_Q.
\end{align}
Remember $\innerpro{f(j',x,c)}{f(j,x,c)}=\delta_{j'j}$. Then
\begin{align}
&\rho^Q_{xc}=tr_P\left(\rho^{(P,Q)}_{xc}\right)\nonumber\\
&=\frac{1}{D}\sum_{j',j\in \intset{D}}\omega^{(j'-j)h(x,c)}tr\left(\outerpro{f(j',x,c)}{f(j,x,c)}_P\right)\nonumber\\
&\outerpro{g(j',x,c)}{g(j,x,c)}_Q\nonumber\\
&=\frac{1}{D}\sum_{j',j\in \intset{D}}\omega^{(j'-j)h(x,c)}\delta_{j'j}\outerpro{g(j',x,c)}{g(j,x,c)}_Q\nonumber\\
&=\frac{1}{D}\sum_{j\in \intset{D}}\outerpro{g(j,x,c)}{g(j,x,c)}_Q.
\end{align}
In this formula functions $h(x,c),f(j,x,c)$ disappear. Thus the attacker cannot obtain any phase-information.\end{proof}
\subsection{Proof of Lemma~\ref{lemma2}}\label{secA.2}
\noindent We first deduce a general upper bound of information disclosure. Follow Definition~\ref{def1} and Lemma~\ref{lemma1}, then we have
\begin{proposition}[\textbf{Upper Bound of Information Disclosure}]\label{proposition1}
Under the attack described in Section~\ref{secA.1},
\begin{align}\label{eq_h}
H(Z&:X)\le \log_2{\abs{S_X}}-\nonumber\\
&\frac{\sum_{b\in\intset{L}}\beta_b\log_2{\beta_b}-\abs{S_X}\sum_{b\in\intset{L}}\alpha_{bx}\log_2{\alpha_{bx}}}{D\abs{S_X}\abs{S_C}},
\end{align} 
where $\alpha_{bx}=\abs{g_x^{-1}(b)}\nonumber=\abs{\set{(j,c)\lvert g(j,x,c)=b}}$ and $\beta_b=\abs{g^{-1}(b)}\nonumber=\abs{\set{(j,x,c)\lvert g(j,x,c)=b}}$.
\end{proposition}
\begin{proof}
As is deduced in Section~\ref{secA.1},
\begin{equation}
\rho^Q_{xc}=\frac{1}{D}\sum_{j\in \intset{D}}\outerpro{g(j,x,c)}{g(j,x,c)}_Q.
\end{equation}
Remember $p_x=\frac{1}{\abs{S_X}}$, $p_c=\frac{1}{\abs{S_C}}$. Then
\begin{align}
&\rho^Q_{x}=\sum_{c\in S_C}\frac{1}{\abs{S_C}}\rho^Q_{xc}\nonumber\\
&=\frac{1}{\abs{S_C}}\sum_{c\in S_C}\frac{1}{D}\sum_{j\in \intset{D}}\outerpro{g(j,x,c)}{g(j,x,c)}_Q\nonumber\\
&=\frac{1}{D\abs{S_C}}\sum_{b\in\intset{L}}\sum_{g(j,x,c)=b}\outerpro{b}{b}_Q\nonumber\\
&=\frac{1}{D\abs{S_C}}\sum_{b\in \intset{L}}\alpha_{bx}\outerpro{b}{b}_Q.
\end{align}
Remember that $g:\intset{D}\times S_X\times S_C\to \intset{L}$. Then we can get its Von Neumann entropy
\begin{align}
&S\left(\rho^Q_{x}\right)=-\sum_{b\in\intset{L}}\frac{\alpha_{bx}}{D\abs{S_C}}\log_2\frac{\alpha_{bx}}{D\abs{S_C}}\nonumber\\
&=\sum_{b\in\intset{L}}\frac{\alpha_{bx}}{D\abs{S_C}}\left(\log_2\left(D\abs{S_C}\right)-\log_2\alpha_{bx}\right)\nonumber\\
&=\frac{\sum_{b\in\intset{L}}\alpha_{bx}}{D\abs{S_C}}\log_2\left(D\abs{S_C}\right)-\frac{1}{D\abs{S_C}}\sum_{b\in\intset{L}}\alpha_{bx}\log_2\alpha_{bx}\nonumber\\
&=\log_2\left(D\abs{S_C}\right)-\frac{1}{D\abs{S_C}}\sum_{b\in\intset{L}}\alpha_{bx}\log_2\alpha_{bx},
\end{align}
where $\sum_{b\in\intset{L}}\alpha_{bx}=\abs{\intset{D}\times S_C}=D\abs{S_C}$. Now
\begin{align}
&\rho^Q=\sum_{x\in S_X}\frac{1}{\abs{S_X}}\rho^Q_{x}\nonumber\\
&=\frac{1}{\abs{S_X}}\sum_{x\in S_X}\frac{1}{\abs{S_C}}\sum_{c\in S_C}\frac{1}{D}\sum_{j\in \intset{D}}\outerpro{g(j,x,c)}{g(j,x,c)}_Q\nonumber\\
&=\frac{1}{D\abs{S_X}\abs{S_C}}\sum_{b\in \intset{L}}\beta_{b}\outerpro{b}{b}_Q.
\end{align}
Similar, its entropy is
\begin{align}
&S\left(\rho^B\right)=-\sum_{b\in\intset{L}}\frac{\beta_{b}}{D\abs{S_X}\abs{S_C}}\log_2\frac{\beta_{b}}{D\abs{S_X}\abs{S_C}}\nonumber\\
&=\log_2\left(D\abs{S_X}\abs{S_C}\right)-\frac{1}{D\abs{S_X}\abs{S_C}}\sum_{b\in\intset{L}}\beta_{b}\log_2\alpha_{b},
\end{align}
where $\sum_{b\in\intset{L}}\beta_{b}=\abs{\intset{D}\times S_X\times S_C}=D\abs{S_X}\abs{S_C}$. Now
\begin{align}
&H(Z:X)\le S\left(\rho^B\right)-\sum_{x\in S_X}p_x S\left(\rho^B_x\right)\nonumber\\
&=\log_2\left(D\abs{S_X}\abs{S_C}\right)-\frac{1}{D\abs{S_X}\abs{S_C}}\sum_{b\in\intset{L}}\beta_{b}\log_2\alpha_{b}-\nonumber\\
&\ \abs{S_X}\frac{1}{\abs{S_X}}\left(\log_2\left(D\abs{S_C}\right)-\frac{1}{D\abs{S_C}}\sum_{b\in\intset{L}}\alpha_{bx}\log_2\alpha_{bx}\right)\nonumber\\
&=\log_2{\abs{S_X}}-\nonumber\\
&\ \ \ \frac{\sum_{b\in\intset{L}}\beta_b\log_2{\beta_b}-\abs{S_X}\sum_{b\in\intset{L}}\alpha_{bx}\log_2{\alpha_{bx}}}{D\abs{S_X}\abs{S_C}}.
\end{align} 
By the Holevo bound\cite{1973HolevoSta}.
\end{proof}

Now we can prove Lemma~\ref{lemma2}.

\begin{proof}[Proof of Lemma~\ref{lemma2}: ]
For convenience, we take the classical information $r_3=r_1^{-1}$, $r_4=c_1-r_2r_3$ as quantum information $\bra{r_3}_{e_1}\bra{r_4}_{e_2}$. Denote $\mathbf{c}=(c_1,c_2,c_3,c_4)$. Now Bob owns system $Q=(t_1,t_2,g,e_1,e_2)$, and we have a function 
\begin{equation}
    g_B(j,x_i,\mathbf{c})=j+c_1\parallel jp_i+c_2 \parallel jc_3+c_4 \parallel r_1^{-1} \parallel c_1-r_2r_3,
\end{equation}
and its value $b=b_1\parallel b_2 \parallel b_3 \parallel r_3 \parallel r_4$, i.e.,
\begin{equation}
\left\{\begin{matrix} 
j+c_1 =b_1 \\
 jp_i+c_2 =b_2\\
jc_3+c_4=b_3 \\
\left(k_1+p_ik_2+c_3k_3\right)^{-1}=r_3 \\
c_1-\left(k_1c_1+k_2c_2+k_3c_4\right)r_3=r_4
\end{matrix}\right.,
\end{equation}
or,
\begin{equation}\label{eq_c}
\left\{\begin{matrix} 
c_1 =b_1-j \\
c_2 =b_2-jp_i\\
jc_3+c_4=b_3 \\
k_3c_3=r_3^{-1}-k_1-p_ik_2 \\
\left(k_1-r_3^{-1}\right)c_1+k_2c_2+k_3c_4=-r_4r_3^{-1}
\end{matrix}\right..
\end{equation}
Take $c_1,c_2,c_3,c_4$ as unknowns. We have its augmented matrix as
\begin{equation}
\left[\begin{array}{c|c}
\begin{matrix} 
1 &0&0&0\\
0&1&0&0\\
0&0&j&1 \\
0&0&k_3&0\\
k_1-r_3^{-1}&k_2&0&k_3
\end{matrix}&
\begin{matrix} 
b_1-j\\
b_2-jp_i\\
b_3 \\
r_3^{-1}-k_1-p_ik_2\\
-r_4r_3^{-1}
\end{matrix}
\end{array}\right].
\end{equation}
It can be Gaussian eliminated to
\begin{align}
&\left[\begin{array}{c|c}
\begin{matrix} 
1 &0&0&0\\
0&1&0&0\\
0&0&1&0 \\
0&0&0&1\\
0&0&0&0
\end{matrix}&
\begin{matrix} 
b_1-j\\
b_2-jp_i\\
k_3^{-1}\left(r_3^{-1}-k_1-p_ik_2\right) \\
b_3-jk_3^{-1}\left(r_3^{-1}-k_1-p_ik_2\right)\\
\left(k_1-r_3^{-1}\right)b_1+k_2b_2+k_3b_3+r_4r_3^{-1}
\end{matrix}
\end{array}\right].
\end{align}
In order for the equation system to have a solution, there must be
\begin{equation}
    \left(k_1-r_3^{-1}\right)b_1+k_2b_2+k_3b_3+r_4r_3^{-1}=0.
\end{equation}
$\forall b_1,b_2,b_3\in \intset{D}, r_3\in\odd{D}$, we have 
\begin{equation}\label{eq_kr}
    r_4=-r_3\left[\left(k_1-r_3^{-1}\right)b_1+k_2b_2+k_3b_3\right],
\end{equation}
thus $\abs{Im(g_B)}=D^3 \cdot \frac{D}{2}=\frac{D^4}{2}$. If (\ref{eq_kr}) is satisfied, then the general solution of (\ref{eq_c}) is 
\begin{equation}
\left\{\begin{matrix} 
c_1 =b_1-j \\
c_2 =b_2-jp_i\\
c_3=k_3^{-1}\left(r_3^{-1}-k_1-p_ik_2\right) \\
c_4=b_3-jk_3^{-1}\left(r_3^{-1}-k_1-p_ik_2\right)
\end{matrix}\right..
\end{equation}
We have $S_{x_i}=\intset{N}$, $S_{\mathbf{c}}=\intset{D}^3\times \odd{D}$. $\forall (j,x_i)\in \intset{D}\times \intset{N}$, we can uniquely identify a solution of $\mathbf{c}$ as above. Therefore, if $b\in Im(g_B)$, then $\alpha_{bx_i}=D$, $\beta_b=DN$; Otherwise, $\alpha_{b{x_i}}=\beta_b=0$. By (\ref{eq_h}), we have
\begin{align}
&H(Z_B:X_A)\le \log_2{\abs{S_{x_i}}}\nonumber\\
&-\frac{\sum_{b\in\intset{L}}\beta_b\log_2{\beta_b}-\abs{S_{x_i}}\sum_{b\in\intset{L}}\alpha_{b{x_i}}\log_2{\alpha_{b{x_i}}}}{D\abs{S_{x_i}}\abs{S_{\mathbf{c}}}}\nonumber\\
&=\log_2{N}-\abs{Im(g_B)}\frac{DN\log_2\left(DN\right)-ND\log_2D}{DN\frac{D^4}{2}}\nonumber\\
&=m-(d+m)+d=0.
\end{align} 
Therefore, we prove that $I_A=H(Z_B:X_A)=0$, i.e., no information about $x_i$ can be stolen.

On the other hand, if particle $g$ is not involved, i.e., $c_3=c_4=0$, then Bob has $k_1+p_ik_2=r_3^{-1}$, and gets $p_i=k_2^{-1}\left(r_3^{-1}-k_1\right)$. Therefore, its existence is necessary.
\end{proof}

\subsection{Proof of Lemma~\ref{lemma3}}\label{secA.3}

\noindent Similar, we first need the following proposition.

\begin{proposition}[Solution of Modular Multiplication Equation]\label{proposition2} 
Let $D=2^d$, where $d$ is a positive integer. Assume there are three random integers $a,b,c\in\intset{D}$, where $a=2^{d_1}w_1,b=2^{d_2}w_2,c=2^{d_3}w_3$, and $d_1,d_2,d_3\in \intset{d+1}, w_i\in\odd{2^{d-d_i}}$. Denote $0=D=2^d\times 1$. We have
\begin{enumerate}[1)]
    \item The necessary and sufficient condition for $ab\equiv c(\mod D)$ is: $d_1+d_2\ge d_3$ if $d_3=d$; $d_1+d_2=d_3$ and $w_1w_2\equiv w_3(\mod 2^{d-d_3})$ if $d_3<d$.
    \item Given $a, c$, select $b$ randomly, denote $p=\Pr\left(ab\equiv c(\mod D)\right)$. Then, $p=\frac{1}{2^{d-d_1}}$ if $d_3\ge d_1$; $p=0$ if $d_3<d_1$.
\end{enumerate}
\end{proposition}
\begin{proof}
\begin{enumerate}[1)]
    \item $ab\equiv c(\mod D)$ is equivalent to $2^{d_1+d_2}w_1w_2=2^{d_3}w_3+k2^d$, where $k$ is an integer. If $d_3=d$, then $w_3=1$, and  $2^{d_1+d_2}w_1w_2=(k+1)2^d$, $d_1+d_2\ge d$; conversely, if $d_1+d_2\ge d$, then $2^{d_1+d_2}w_1w_2=(k+1)2^d\equiv c(\mod D)$. If $d_3<d$, then $2^{d_1+d_2}w_1w_2=2^{d_3}\left(w_3+k2^{d-d_3}\right)$, where $w_3+k2^{d-d_3}$ is odd. Therefore ,we have $d_1+d_2=d_3$, $w_1w_2=w_3+k2^{d-d_3}$, i.e., $w_1w_2\equiv w_3(\mod 2^{d-d_3})$.
    \item {a) If $d_3=d$, then 
        \begin{align}
            &p_1=\Pr\left(ab\equiv c(\mod D)\lvert d=d_3\right)\nonumber\\
            &=\Pr\left(d_1+d_2\ge d_3\right)=\Pr\left(d_2\ge d-d_1\right)\nonumber\\
            &=\Pr\left(2^{d-d_1}\lvert b\right)=\frac{2^{d-(d-d_1)}}{2^d}=\frac{1}{2^{d-d_1}}.
        \end{align}
    Note that there are $2^{d-(d-d_1)}$ multiples of $2^{d-d_1}$ in $\intset{D}$.
    
    b) If $d>d_3\ge d_1$, the conditions are $d_2=d_3-d_1$ and $w_2\equiv w_1^{-1}w_3(\mod 2^{d-d_3})$, where $w_1^{-1}\in \odd{2^{d-d_3}}$ is the multiplicative inverse of $w_1\mod 2^{d-d_3}$. First, $d_2=d_3-d_1$ means that $2^{d_3-d_1}\lvert b$ and $2^{d_3-d_1+1}\nmid b$. Of the $2^{d-(d_3-d_1)}$ multiples of $2^{d_3-d_1}$, half are multiples of $2^{d_3-d_1+1}$, thus
        \begin{align}
            &p_{2a}=\Pr\left(d_2=d_3-d_1\right)\nonumber\\
            &=\frac{2^{d-(d_3-d_1)}2^{-1}}{2^d}=\frac{1}{2^{d_3-d_1+1}},
        \end{align}
    If $d_2=d_3-d_1$, we need 
    \begin{equation}
        w_2=w_1^{-1}w_3+k2^{d-d_3}\in\odd{2^{d-d_3+d_1}},
    \end{equation}
    i.e., $1\le w_1^{-1}w_3+k2^{d-d_3}\le 2^{d-d_3+d_1}-1$. No matter how much $w_1^{-1}w_3$ equals, we can assume that
    \begin{equation}
        l2^{d-d_3}+1\le w_1^{-1}w_3\le (l+1)2^{d-d_3}-1,
    \end{equation}
    where $l$ is an integer. Then, we can deduce that
    \begin{align}
        &-(l+1)2^{d-d_3}+1\le -w_1^{-1}w_3\le k2^{d-d_3}\nonumber\\
        &\le 2^{d-d_3+d_1}-1-w_1^{-1}w_3\le 2^{d-d_3+d_1}-1 -l2^{d-d_3}-1\nonumber\\
        &=2^{d-d_3}(2^{d_1} -l)-2,
    \end{align}
    thus $-l\le k\le 2^{d_1} -l-1$, i.e., $k$ has $2^{d_1}$ possible values. Then 
        \begin{align}
        &p_{2b}=\Pr\left(w_2=w_1^{-1}w_3+k2^{d-d_3}\lvert d_2=d_3-d_1\right)\nonumber\\
            &=\frac{2^{d_1}}{2^{d-d_3+d_1-1}}=\frac{1}{2^{d-d_3-1}}.
        \end{align}
    Note that there are half odd integers in $\intset{2^{d-d_3+d_1}}$. Now
        \begin{align}
            &p_{2}=\Pr\left(ab\equiv c(\mod D)\lvert d>d_3\ge d_1\right)=p_{2a}\cdot p_{2b}\nonumber\\
            &=\frac{1}{2^{d_3-d_1+1}}\cdot \frac{1}{2^{d-d_3-1}}=\frac{1}{2^{d-d_1}}=p_{1}.
        \end{align}
    Thus $p=p_1=p_2=\frac{1}{2^{d-d_1}}$.
    
    c) If $d_3<d_1$, then $d_1+d_2\ne d_3$, thus $p=0$.
    }
\end{enumerate}
In total, we prove this proposition.
\end{proof}

Now we can prove Lemma~\ref{lemma3}.

\begin{proof}[Proof of Lemma~\ref{lemma3}:]
There are three possible cases:
    \begin{enumerate}[a)]
        \item {Using $t_1$ to steal $s_i$ and $t_2$ to steal $q_i$ simultaneously. 
        
        Assume that Alice prepares $\sfrac{D}\sum_{j\in\intset{D}}\bra{j}_{t_1}, \sfrac{D}\sum_{j\in\intset{D}}\bra{j}_{t_2}$.
        For convenience, we set $h,g$ in $\bra{0}$, since they are used to protect herself, not Bob. In Step 2 of Operation stage, the state will be
         \begin{align}
&\frac{1}{D}\sum_{j,j'\in\intset{D}}\bra{j}_{t_1}\bra{j'}_{t_2}\bra{0}_{g}\rightarrow\nonumber\\
&\frac{1}{D}\sum_{j,j'\in\intset{D}}\bra{jk}_{t_1}\bra{j'}_{t_2}\bra{jk_1+j'k_2}_{g}.
        \end{align}
        Now $t_1,t_2,g$ are entangled. In Step 4 of Operation stage, no matter what value $r_3,r_4$ she sends are, the total state will be 
        \begin{align}
&\frac{1}{D}\sum_{j,j'\in\intset{D}}\omega^{js_i+j'q_i}\bra{j}_{t_1}\bra{j'}_{t_2}\bra{jk_1+j'k_2}_{g}\rightarrow\nonumber\\
&\sfrac{D}\sum_{j\in\intset{D}}\omega^{js_i}\bra{j}_{t_1}\sfrac{D}\sum_{j'\in\intset{D}}\omega^{j'q_i}\nonumber\\
&\bra{j'}_{t_2}\bra{(jk_1r_3+j'k_2r_3+r_4)\oplus j}_{g}.
        \end{align}
        Now Bob will measure $g$. For each $j\in\intset{D}$, we have
        \begin{align}
        &\Pr\left((jk_1r_3+j'k_2r_3+r_4)\oplus j=0\right)\nonumber\\
        &=\Pr\left(jk_1r_3+j'k_2r_3+r_4=j\right)\nonumber\\
        &=\Pr\left(j'=(k_2r_3)^{-1}(j(1-k_1r_3)-r_4)\right)=\frac{1}{D}.
        \end{align}
        Thus, he will find Alice's attack easily, with probability $1-\frac{1}{D}$.
        }
        \item {Using $t_1$ to steal $s_i$ only. 
        
        Similarly, Alice prepares $\sfrac{D}\sum_{j\in \intset{D}}\bra{j}_{t_1},\bra{0}_{t_2},\bra{0}_{g}$.
        In Step 2 of Operation stage, it will be
        \begin{align}
            \sfrac{D}\sum_{j\in \intset{D}}\bra{j}_{t_1}\bra{0}_{t_2}\bra{0}_{g}\rightarrow \sfrac{D}\sum_{j\in \intset{D}}\bra{j}_{t_1}\bra{0}_{t_2}\bra{jk_1}_{g},
        \end{align}
        and in Step 4 of Operation stage, it will be
        \begin{align}
            &\sfrac{D}\sum_{j\in \intset{D}}\omega^{js_i}\bra{j}_{t_1}\bra{jk_1}_{g}\rightarrow \nonumber\\
            &\sfrac{D}\sum_{j\in \intset{D}}\omega^{js_i}\bra{j}_{t_1}\bra{(jk_1r_3+r_4)\oplus j}_{g}.
        \end{align}
        When Bob measure $g$, he will get $\bra{0}$ if $jk_1r_3+r_4=j$, i.e., $j(1-k_1r_3)=r_4$. Let $S_J=\{j|j(1-k_1r_3)=r_4,j\in\intset{D}\}$. Bob can find the cheating with probability $1-\frac{\abs{S_J}}{D}$, and ignore it with $\frac{\abs{S_J}}{D}$. If he ignores, the state of $t_1$ will be
        \begin{equation}
    \bra{\psi_{s_i}}=\sfrac{\abs{S_J}}\sum_{j\in S_J}\omega^{js_i}\bra{j}.
        \end{equation}
        $t_1$ will return to Alice in Step 5 of Operation stage. Now the density operator is 
        \begin{equation}
        \begin{aligned}
\rho_{s_i}=\outerpro{\psi_{s_i}}{\psi_{s_i}}=\frac{1}{\abs{S_J}}\sum_{j',j\in S_J}\omega^{(j'-j)s_i}\outerpro{j'}{j}
        \end{aligned}
        \end{equation}
        Since $\rho_{s_i}=\outerpro{\psi_{s_i}}{\psi_{s_i}}$, $S\left(\rho_{s_i}\right)=0$. For Alice who has no knowledge of odd integer $s_i=4v_i-2y_i-1\mod D$, $s_i$ can be considered to follow the uniform distribution on $\odd{D}$, i.e., $p_{s_i}=\frac{2}{D}$. Let $s_i=2a+1$, then the total operator is
        \begin{align}\label{eq_rho}
&\rho=\sum_{s_i\in\odd{D}}\frac{2}{D}\rho_{s_i}\nonumber\\
&=\frac{2}{D\abs{S_J}}\sum_{j\in S_J}\sum_{j'\in S_J}\sum_{s_i\in\odd{D}}\omega^{(j'-j)s_i}\outerpro{j'}{j}.
        \end{align}
        Note that if $j'=j$ then $\sum_{s_i\in\odd{D}}\omega^{(j'-j)s_i}=\frac{D}{2}$;
        Or, if $j'-j\equiv \frac{D}{2}(\mod D)$, then
        \begin{equation}
        \begin{aligned}
\sum_{s_i\in\odd{D}}\omega^{(j'-j)s_i}=\sum_{s_i\in\odd{D}}e^{\imath \pi s_i}=\sum_{s_i\in\odd{D}}(-1)=-\frac{D}{2};
        \end{aligned}
        \end{equation}
        Otherwise, 
        \begin{align}
&\sum_{s_i\in\odd{D}}\omega^{(j'-j)s_i}
=\sum_{a\in\intset{\frac{D}{2}}}e^{\frac{\imath 2\pi}{D}\cdot(j'-j)(2a+1)}\nonumber\\
&=\omega^{j'-j}\sum_{a\in\intset{\frac{D}{2}}}e^{\frac{\imath 2\pi}{D}\cdot2(j'-j)a}=\omega^{j'-j}\frac{1-e^{\imath 2\pi(j'-j)}}{1-e^{\frac{\imath 2\pi}{D}\cdot2(j'-j)}}\nonumber\\
&=0.
        \end{align}
        Since $k_1r_3$ is odd, we have $D|(1-k_1r_3)\frac{D}{2}$, thus
        \begin{align}
(1-k_1r_3)\left(j+\frac{D}{2}\right)\equiv (1-k_1r_3)j\equiv r_4(\mod D).
        \end{align}
        I.e., if $j\in S_J$ then $j+\frac{D}{2}\mod D\in S_J$. Now we can calculate the eigenvalue of $\rho$. By (\ref{eq_rho}), we have
        \begin{equation}
\rho=\frac{1}{\abs{S_J}}\sum_{j\in S_J}\left(\bra{j}-\bra{j+\frac{D}{2}}\right)\ket{j}.
        \end{equation}
        $\forall j\in S_J$, denote $\bra{\phi_{j\pm}}=\bra{j}\pm\bra{j+\frac{D}{2}}$, then
        \begin{equation}
\rho\bra{\phi_{j+}}=0\bra{\phi_{j+}},\rho\bra{\phi_{j-}}=\frac{2}{\abs{S_J}}\bra{\phi_{j+}}.
        \end{equation}
        I.e., $\rho$ has $\frac{\abs{S_J}}{2}$ non-zero eigenvalues $\frac{2}{\abs{S_J}}$. Then
        \begin{equation}
S\left(\rho\right)=-\frac{\abs{S_J}}{2}\cdot\frac{2}{\abs{S_J}}\log_2 \frac{2}{\abs{S_J}}=\log_2\abs{S_J}-1.
        \end{equation}
        If Alice measures $t_1$ to obtain any result $Z_A$, then
        \begin{align}
        &H\left(Z_A:s_i|\rm{Alice\ passes\ the\ test}\right)\nonumber\\
        &\le S\left(\rho\right)-\sum_{s_i\in \odd{D}}p_{s_i}S\left(\rho_{s_i}\right)=\log_2\abs{S_J}-1
        \end{align}
        Since the probability she pass Bob's honesty test is $\frac{\abs{S_J}}{D}$, then the leakage degree of Bob's privacy, i.e., the average information Alice obtains is
        \begin{equation}\label{eq_I}
        \begin{aligned}
        I_B=H(Z_A:s_i)\le \sum_{\abs{S_J}}\frac{\abs{S_J}}{D}p_{\abs{S_J}}\log_2\abs{S_J}-1.
        \end{aligned}
        \end{equation}
        Let's find $p_{\abs{S_J}}$. Let $1-k_1k_3=2^{d_l}w_l$, $r_4=2^{d_r}w_r$, $j=2^{d_j}w_j$. By Proposition~\ref{proposition2}, $j(1-k_1r_3)=r_4$ only if $d_r\ge d_l$, with probability $\frac{1}{2^{d-d_l}}$, i.e., $\abs{S_J}=\frac{2^d}{2^{d-d_l}}=2^{d_l}$; Otherwise, if $d_r<d_l$, then $\abs{S_J}=0$. To maximize $I_B$, Alice should set $d_r=d$, i.e., $r_4=0$, so that $d_r\ge d_l$ always holds. Note that $1-k_1r_3$ is a random even, if $d_l<d$, then
        \begin{align}
            &p_{\abs{S_J}}=\Pr\left(2^{d_l}\lvert(1-k_1r_3),\ 2^{d_l+1}\nmid(1-k_1r_3)\right)\nonumber\\
            &=\frac{2^{d-d_l-1}}{2^{d-1}}=\frac{1}{2^{d_l}};
        \end{align}
        If $d_l=d$, then $1-k_1r_3=0$, and $p_{\abs{S_J}}=\frac{2}{2^{d-1}}$. Therefore,
        \begin{equation}
        p_{\abs{S_J}}=\left\{
        \begin{matrix}
        \frac{1}{\abs{S_J}},&\ \mathrm{if}\ \abs{S_J}< D\\
        \frac{2}{\abs{S_J}},&\ \mathrm{if}\ \abs{S_J}=D
        \end{matrix}\right..
        \end{equation}
        Now by (\ref{eq_I}), we have 
        \begin{align}
        &I_B\le \sum_{\abs{S_J}}\frac{\abs{S_J}}{D}p_{\abs{S_J}}\log_2\abs{S_J}-1\nonumber\\
        &=\sum_{d_l=1}^{d-1}\frac{2^{d_l}}{D}\cdot\frac{1}{2^{d_l}}\left(\log_2 2^{d_l}-1\right)+\frac{2^d}{D}\cdot\frac{2}{2^d}\left(\log_2 2^d-1\right)\nonumber\\
        &=\frac{1}{D}\sum_{d_l=1}^{d-1}(d_l-1)+\frac{2}{D}(d-1)=\frac{1}{D}\frac{(d-1)(d-2)+4(d-1)}{2}\nonumber\\
        &=\frac{1}{D}\frac{(d-1)(d+2)}{2}=\frac{d^2+d-2}{2\cdot2^d}=O\left(\frac{d^2}{2^d}\right).
        \end{align}
        }
        
        \item {Using $t_2$ to steal $q_i$ only.
        
        This case is almost identical to b). We can calculate the total density operator $\rho=\sum_{y_i\in\intset{N}}\frac{1}{N}\rho_y$, and find that the upper bound in this case equals to the one in b) if $1-k_2r_3\not\equiv 0(\mod D)$, because if $j-j'\equiv\mathrm{odd}(\mod D)$, then $j,j'$ cannot satisfy $j(1-k_2r_3)\equiv r_4(\mod D)$ at the same time. In this case, we have $I_B=H(Z_A:s_i)=O\left(\frac{d^2}{2^d}\right)$.

        If $1-k_2r_3\equiv 0(\mod D)$ with probability $\frac{2}{D}$, we can deduce a new upper bound $d$, since at most $d$ bits of classical information can be extracted from the $d$-qubit state\cite{2000NielsenQua}. Considering the probability $\frac{1}{2^{d-1}}$, the average upper bound is
        \begin{equation}
        \begin{aligned}
            I_B\le \left( 1-\frac{1}{2^{d-1}}\right)O\left(\frac{d^2}{2^d}\right)+\frac{1}{2^{d-1}}d\le O\left(\frac{d^2}{2^d}\right).
        \end{aligned}
        \end{equation}
        }
    \end{enumerate}
    In total, the leakage degree of Bob's privacy is $I_B=O\left(\frac{d^2}{2^d}\right)$, where the asymptotic coefficient is $\frac{1}{2}$.
\end{proof}

\subsection{Proof of Lemma~\ref{lemma4}}\label{secA.4}
\begin{proof}[]Denote random variables $M=(M_1,M_2,\cdots,M_n)$ and $X_B=(\mathbf{y},v)$. At first, we have\cite{1978WynerDef}
\begin{align}
I_B=H(M:X_B|u)=H(X_B:M,u)-H(X_B:u).
\end{align} 
Since $u$ is a function of $X_B$ (and $M$, since Alice can calculate $u$ only with $M$), we have $H(X_B:u)=H(u)$, $H(u|X_B)=0$ and $H(M,u)=H(M)$ (Theorem 11.3 of Ref.~\cite{2000NielsenQua}). By the chaining rule for conditional entropies (Theorem 11.4 of Ref.~\cite{2000NielsenQua}), $H(M,u|X_B)=H(u|M,X_B)+H(M|X_B)=H(M,X_B)-H(X_B)$. Then
\begin{align}
&I_B=H(M:X_B|u)=H(X_B:M,u)-H(X_B:u)\nonumber\\
&=H(M,u)-H(M,u|X_B)-H(u)\nonumber\\
&=H(M)-H(M,X_B)+H(X_B)-H(u)\nonumber\\
&=H(X_B)-H(X_B|M)-H(u).
\end{align}
Note that $\hat{M}_i\equiv M_i-2x_i\equiv 4v_i+4x_iy_i(\mod D)$ and $M_i$ correspond one-to-one, then $H(X_B|\hat{M})=H(X_B|M)$, where $\hat{M}=(\hat{M}_1,\cdots,\hat{M}_n)$. $\forall \hat{M_i}\in 4\intset{N},y_i\in \intset{N}$, only if $v_i$ satisfy $4v_i\equiv \hat{M_i}-4x_iy_i(\mod D)$, $4v\equiv 4\sum_{i=1}^n v_i(\mod D)$ is valid. Therefore, there is only one possible $4v\in\intset{D}$, and thus one $v\in\intset{N}$. Then we have $p(X_B|\hat{M})=\frac{1}{N^{n}}$. On the other hand, $\forall y_i,v\in \intset{N}$, if $v_1,v_2,\cdots,v_{n-1}\in\intset{N}$ have been selected, then $4v_n\equiv 4v-4\sum_{i=1}^{n-1} v_i(\mod D)$ is determined, and so is $\hat{M}$. Therefore, $p(\hat{M}|X_B)=\frac{1}{N^{n-1}}$. By $p(X_B)=\frac{1}{N^{n+1}}$, we have
\begin{align}
&H(X_B|\hat{M})=-\sum_{X_B\in\intset{N}^{n+1}}\sum_{\hat{M}\in 4\intset{N}^n}p(X_B,\hat{M})\log_2{p(X_B|\hat{M})}\nonumber\\
&=-\sum_{X_B\in\intset{N}^{n+1}}\sum_{v_1,v_2,\cdots,v_{n-1}\in\intset{N}}p(\hat{M}|X_B)p(X_B)\log_2{p(X_B|\hat{M})}\nonumber\\
&=-{N}^{n+1}{N}^{n-1}\frac{1}{N^{n-1}}\frac{1}{N^{n+1}}\log_2{\frac{1}{N^n}}=nm.
\end{align}
Similarly, because $\forall u\in\intset{N}$, we can randomly select $\mathbf{y}\in\intset{N}^n$, then $v=u-\mathbf{x}\cdot\mathbf{y}$. That is, each $u$ corresponds to $N$ possible values of $y$, i.e., $p(u)=\frac{1}{N}$. Then $H(X_B)=(n+1)m$, $H(u)=m$, and $I_B=(n+1)m-nm-m=0$.
\end{proof}

\section{Implementation of Quantum Gates}\label{secB}
\noindent Here we analyze the complexity of the quantum operations we used, and how to implement them in simulation. 

\subsection{Complexity Analysis}\label{secB.1}
\noindent
\begin{enumerate}[1)]
\item As shown in figure~\ref{fig_QFT}, A $\QFT$ (or $\QFIT$) gate can be decomposed into $O\left(d^2\right)$ $\H$ and $\CPROT$ gates\cite{2000NielsenQua}.

\item {Let $a=\sum_{i\in \intset{d}} a_i2^i$, $b=\sum_{k\in \intset{d}} b_k2^k$. As shown in figure~\ref{fig_ROT}, we have $\ROT_h=\prod_{i,k\in \intset{d}}\PROT(i+k)^{b_i}_{h_i}$, where $b_i$ means the controlling bit, i.e., if $b_i=1$, perform $\PROT(i+k)_{h_i}$, otherwise do nothing. Because
    \begin{align}
        &\prod_{i,k\in \intset{d}}\PROT(i+k)^{b_i}_{h_i}\bra{a}_{h}=\prod_{i,k\in \intset{d}}e^\frac{\imath 2\pi 2^{i+k}a_ib_i}{D}\bra{a}_{h}\nonumber\\
        &=e^\frac{\imath 2\pi \sum_{i,k\in \intset{d}}2^{i+k}a_ib_i}{D}\bra{a}_h=e^\frac{\imath 2\pi ab}{D}\bra{a}_h\nonumber\\
        &=\ROT_h\bra{a}_h.
    \end{align}
Thus, the complexity of $\ROT$ gate is also $O\left(d^2\right)$.}

\item As elementary arithmetic gates, the complexity of $\SUM$ and $\BSUM$ are both $O(d)$, if we only use $\Toffoli$ and $\CNOT$ gates corresponding to classical AND and XOR gate. Generally, it requires an auxiliary qubit as a carrier.

\item {To realize $\MUL$, introduce an auxiliary register $\bra{0}_t$, then
    \begin{align}
&\bra{a}_h\bra{0}_t\onarrow{\SUM(2^0b)^{h_0}_t}\bra{a}_h\bra{2^0a_0b}_t\onarrow{\SUM(2^1b)^{h_1}_t}\nonumber\\
        &\cdots\onarrow{\SUM(2^{d-1}b)^{h_{d-1}}_t}=\bra{a}_h\bra{ab}_t,
    \end{align}
where $\SUM^{h_i}$ means $\SUM$ is controlled by qubit $h_i$. We denote the above gate as $\BMUL(b)_{(h,t)}$. To realize $\MUL$, use $\SWAP_{(h,t)}=\XOR_{(h,t)}\XOR_{(t,h)}\XOR_{(h,t)}: \bra{a}_h\bra{b}_t\rightarrow \bra{b}_h\bra{a}_t$. Then
    \begin{align}
&\bra{a}_h\bra{0}_t\onarrow{\BMUL(b)_{(h,t)}}\bra{a}_h\bra{ab}_t\onarrow{\BMUL(-b^{-1})_{(t,h)}}\nonumber\\
&\bra{0}_h\bra{ab}_t\onarrow{\SWAP_{(h,t)}}\bra{ab}_h\bra{0}_t,
    \end{align}
as Shor described\cite{1994ShorAlg}. Its complexity is $O(d^2)$.}

\item Obviously, an $\XOR$ gate can decomposed into $d$ $\CNOT$ gates, as shown in figure~\ref{fig_XOR}. Thus its complexity is $O(d)$. 
\end{enumerate}

\subsection{Implementation of Quantum Gates}\label{secB.2}
\noindent To save valuable qubit resources, we use Draper's Transform Adder\cite{2000DraperAdd} to implement $\SUM$ and $\BSUM$. It need no auxiliary qubits, but increases the complexity. In this method, we have $\SUM(b)=\QFIT \ROT(b)\QFT$ and $\BSUM_{(h,t)}=\QFIT_t \BROT_{(h,t)}\QFT_t$, with complexity $O(d^2)$. The gate $\BROT$ is defined as $\BROT_{(h,t)}:\bra{a}_h\bra{b}_t\rightarrow \omega^{ab}\bra{a}_h\bra{b}_t$, by replacing $\PROT(i+k)^{b_i}_{h_i}$ with $\CPROT(i+k)_{(h_i,t_i)}$ in $\ROT_h$. The circuit of them are shown in figure~\ref{fig_SUM} and \ref{fig_BSUM} respectively.

In addition, the general $\MUL$ gate needs $O(d)$ auxiliary qubits. Here we point out that for the special modulus $D=2^d$, we can design a special circuit which requires no auxiliary qubits. To illustrate our design idea, we first define a new gate $\SUM(b\mod 2^l)$ as a $\mod 2^l$ summation gate on register $(h_{d-1},h_{d-2},\cdots,h_{d-l})$. E.g., $\SUM(b\mod 2^2)$ will change $\bra{a_{d-1}}\bra{a_{d-2}}$ to $\bra{\sum_{i\in\intset{2}}2^ia_{i+d-2} +b \mod 2^2}$. Denote $w=\frac{b-1}{2}$ (must be an integer). We have
    \begin{align}\label{eq_mul_decom}
        &\MUL(b)=\SUM(w\mod 2^{d-1})^{h_{0}}\SUM(w\mod 2^{d-2})^{h_{1}}\nonumber\\
        &\cdots\SUM(w\mod 2^2)^{h_{d-3}}\SUM(w\mod 2^1)^{h_{d-2}},
    \end{align}
as shown in figure~\ref{fig_MUL}. Now we prove its correctness.

\begin{proof} We have
    \begin{align}
            &ab\equiv a(2w+1)\equiv a+2w\sum_{i=0}^{d-1}2^{i}a_i\equiv a+\sum_{i=1}^{d}2^{i}wa_{i-1}(\mod 2^d)\nonumber\\
            &\equiv a+\sum_{i=1}^{d-1}2^{i}wa_{i-1}(\mod 2^d).
    \end{align}
    Let $a^{(0)}=a$ and $a^{(l+1)}=a^{(l)}+wa_{d-1-l}2^{d-l} \mod D$, then we have $a^{(d)}=a^{(d-1)}+wa_{0}2^{0} \mod D=ab$.
    Let $a^{(l+1)}=\sum_{i=0}^{d-1}a_i^{(l+1)}2^i$, $a^{(l)}=\sum_{i=0}^{d-1}a_i^{(l)}2^i$, then
        \begin{align}
        &\sum_{i=0}^{d-1}a_i^{(l+1)}2^i=a^{(l)}+wa_{d-1-l}2^{d-l} \mod D\nonumber\\
        &=\sum_{i=0}^{d-1}a_i^{(l)}2^i+wa_{d-1-l}2^{d-l} \mod D\nonumber\\
        &=\sum_{i=0}^{d-l-1}a_i^{(l)}2^i+\sum_{i=d-l}^{d-1}a_i^{(l)}2^i+wa_{d-1-l}2^{d-l}-k2^d\nonumber\\
        &=\sum_{i=0}^{d-l-1}a_i^{(l)}2^i+2^{d-l}\left(\sum_{i=0}^{l-1}a_{i+d-l}^{(l)}2^l+wa_{d-1-l}-k2^l\right)\nonumber\\
        &=\sum_{i=0}^{d-l-1}a_i^{(l)}2^i\nonumber\\
        &\ \ \ \ +2^{d-l}\left(\sum_{i=0}^{l-1}a_{i+d-l}^{(l)}2^l+wa_{d-1-l}\mod 2^l\right),
        \end{align}
    where $k$ is integer. Note that $a_i^{(l+1)}=a_i^{(l)}$, $0\le i \le d-l-1$, thus 
    \begin{align}
        &\sum_{i=0}^{d-l-1}a_i^{(l+1)}2^i=\sum_{i=0}^{d-l-1}a_i^{(l)}2^i=\nonumber\\
        &\cdots=\sum_{i=0}^{d-l-1}a_i^{(0)}2^i=\sum_{i=0}^{d-l-1}a_i2^i.
    \end{align}
    Therefore,
    \begin{align}
        &a^{(l+1)}=\sum_{i=0}^{d-l-1}a_i^{(l)}2^i\nonumber\\
        &\ \ \ \ +2^{d-l}\left(\sum_{i=0}^{l-1}a_{i+d-l}^{(l)}2^l+wa_{d-1-l}^{(l)}\mod 2^l\right),
    \end{align}
    which means that we can add $w$ on $\bra{a_{d-1}}\cdots\bra{a_{d-l}}\mod 2^l$ if $a_{d-l-1}=1$, for each $l=0,1,\cdots,d-1$. That's why (\ref{eq_mul_decom}) and the circuit in figure~\ref{fig_MUL} are correct.
\end{proof}

\end{document}